\documentclass[acmsmall]{acmart}

\acmVolume{1}
\acmNumber{} 
\acmArticle{1}
\acmYear{}
\acmMonth{1}
\acmDOI{} 
\startPage{1}

\setcopyright{none}

\usepackage{mathtools}
\usepackage{galois}
\usepackage{tikz}

\newtheorem{clm}{Claim}[section]

\usepackage[normalem]{ulem}

\begin{document}
	
	\title{Abstract interpretation, Hoare logic, and incorrectness logic for quantum programs}     
	

\author{Yuan Feng}
\orcid{nnnn-nnnn-nnnn-nnnn}             
\affiliation{
	\position{Professor}
	\department{Centre for Quantum Software and Information}              
	\institution{University of Technology Sydney}            
	\state{NSW}
	\country{Australia}                    
}

\email{Yuan.Feng@uts.edu.au}          

\author{Sanjiang Li}
\orcid{nnnn-nnnn-nnnn-nnnn}             
\affiliation{
	\position{Professor}
	\department{Centre for Quantum Software and Information}              
	\institution{University of Technology Sydney}            
	\state{NSW}
	\country{Australia}                    
}

\email{Sanjiang.Li@uts.edu.au}          




	\begin{abstract}
			Abstract interpretation, Hoare logic, and incorrectness (or reverse Hoare) logic are powerful techniques for static analysis of computer programs. All of them have been successfully extended to the quantum setting, but largely developed in parallel. In this paper, we examine the relationship between these techniques in the context of verifying quantum while-programs, where the abstract domain and the set of assertions for quantum states are well-structured. In particular, we show that any complete quantum abstract interpretation induces a  quantum Hoare logic and a quantum incorrectness logic, both of which are sound and relatively complete. Unlike the logics proposed in the literature, the induced logic systems are in a forward manner, making them more useful in certain applications. 
			Conversely, any sound and relatively complete quantum Hoare logic or quantum incorrectness logic induces a complete quantum abstract interpretation. As an application, we are able to show the non-existence of any sound and relatively complete quantum Hoare logic or incorrectness logic if tuples of local subspaces are taken as assertions.
	\end{abstract}

\ccsdesc[500]{Theory of computation~Denotational semantics}
\ccsdesc[500]{Theory of computation~Axiomatic semantics}
\ccsdesc[500]{Theory of computation~Hoare logic}
\ccsdesc[500]{Theory of computation~Program verification}
\ccsdesc[500]{Theory of computation~Assertions}

\keywords{Quantum programming, abstract interpretation, incorrectness logic}  

	\maketitle

\newcommand {\empstr} {\Lambda}

\newcommand {\qcf}[1] {{\sf{#1}}}

\newcommand {\qc}[1] {{\sf{#1}}}
\def\>{\ensuremath{\rangle}}
\def\<{\ensuremath{\langle}}
\def\sl {\ensuremath{\llparenthesis}}
\def\sr{\ensuremath{\rrparenthesis}}
\def\-{\ensuremath{\textrm{-}}}
\def\ott{t}
\def\otu{u}
\def\ots{s}
\def\apply{\mathrel{*\!\!=}}

\def\comm{\ensuremath{\leftrightarrow^*}}
\def\reach{\ensuremath{\rightarrow^*}}

\def\ctp{P}
\def\ctq{Q}

\def\change{\ensuremath{\mathit{change}}}

\def\qVar{\ensuremath{\mathit{qv}}}
\def\qv{\ensuremath{\mathit{qv}}}
\def\cVar{\ensuremath{\mathit{cv}}}
\def\QVar{\ensuremath{\mathit{V}}}
\def\CVar{\ensuremath{\mathit{cVar}}}
\def\Var{\ensuremath{\mathit{var}}}
\def\Chan{\mathit{chan}}
\def\cChan{\mathit{cChan}}
\def\qChan{\mathit{qChan}}
\def\BExp{\mathit{BExp}}
\def\Exp{\mathit{Exp}}

\def\fdmu{\Delta}
\def\fdnu{\dnu}
\def\fdomega{\domega}

\def\dmu{\mu}
\def\dnu{\nu}
\def\domega{\omega}
\def\expect{\mathbb{E}}
\def\preexpect{\mathrm{pre}\mathbb{E}}

\def\rassign{:=_{\$}}
\def\fpi{\widehat{\pi}}
\def\h{\ensuremath{\mathcal{H}}}
\def\p{\ensuremath{\mathcal{P}}}
\def\l{\ensuremath{\mathcal{L}}}
\def\g{\ensuremath{\mathcal{G}}}
\def\lh{\ensuremath{\mathcal{L(H)}}}
\def\dh{\ensuremath{\mathcal{D(H})}}
\def\dhv{\ensuremath{\d(\h_V)}}
\def\shv{\ensuremath{\s(\h_V)}}
\def\q{\bold Q}
\def\Q{\ensuremath{\mathbb Q}}
\def\P{\ensuremath{\mathbb P}}
\def\SO{\ensuremath{\mathcal{SO}}}
\def\HP{\ensuremath{\mathcal{HP}}}
\def\hpe{\ensuremath{\mathcal{\e}}}

\def\r{\ensuremath{\mathcal{R}}}
\def\R{\ensuremath{\mathbb{R}}}
\def\m{\ensuremath{\mathcal{M}}}
\def\u{\ensuremath{\mathcal{U}}}
\def\k{\ensuremath{\mathcal{K}}}
\def\K{\ensuremath{\mathfrak{K}}}
\def\S{\ensuremath{\mathfrak{S}}}
\def\s{\ensuremath{\mathcal{S}}}
\def\t{\ensuremath{\mathcal{T}}}
\def\u{\ensuremath{\mathcal{U}}}
\def\U{\ensuremath{\mathfrak{U}}}
\def\L{\ensuremath{\mathfrak{L}}}
\def\x{\ensuremath{\mathcal{X}}}
\def\y{\ensuremath{\mathcal{Y}}}
\def\z{\ensuremath{\mathcal{Z}}}
\def\v{\ensuremath{\mathcal{V}}}

\def\st{\ensuremath{\mathfrak{t}}}
\def\su{\ensuremath{\mathfrak{u}}}
\def\ss{\ensuremath{\mathfrak{s}}}

\def\ra{\ensuremath{\rightarrow}}
\def\a{\ensuremath{\mathcal{A}}}
\def\b{\ensuremath{\mathcal{B}}}
\def\c{\ensuremath{\mathcal{C}}}

\def\e{\ensuremath{\mathcal{E}}}
\def\f{\ensuremath{\mathcal{F}}}
\def\l{\ensuremath{\mathcal{L}}}
\def\X{\mbox{\bf{X}}}
\def\N{\mathbb{N}}
\def\sreal{\mathbb{R}}
\def\Z{\mathbb{Z}}

\def\qzz{\ensuremath{|0\>_q\<0|}}
\def\qoo{\ensuremath{|1\>_q\<1|}}
\def\qzo{\ensuremath{|0\>_q\<1|}}
\def\qoz{\ensuremath{|1\>_q\<0|}}
\def\qii{\ensuremath{|i\>_q\<i|}}
\def\qiz{\ensuremath{|i\>_q\<0|}}
\def\qzi{\ensuremath{|0\>_q\<i|}}

\def\quzz{\ensuremath{|0\>_{\bar{q}}\<0|}}
\def\quoo{\ensuremath{|1\>_{\bar{q}}\<1|}}
\def\quzo{\ensuremath{|0\>_{\bar{q}}\<1|}}
\def\quoz{\ensuremath{|1\>_{\bar{q}}\<0|}}
\def\quii{\ensuremath{|i\>_{\bar{q}}\<i|}}
\def\quiz{\ensuremath{|i\>_{\bar{q}}\<0|}}
\def\quzi{\ensuremath{|0\>_{\bar{q}}\<i|}}

\DeclarePairedDelimiter{\ceil}{\lceil}{\rceil}

\def\d{\ensuremath{\mathcal{D}}}
\def\dh{\ensuremath{\mathcal{D(H)}}}
\def\lh{\ensuremath{\mathcal{L(H)}}}
\def\le{\ensuremath{\sqsubseteq}}
\def\ge{\ensuremath{\sqsupseteq}}
\def\eval{\ensuremath{{\psi}}}
\def\aeq{\ensuremath{{\ \equiv\ }}}
\def\osnt{\ensuremath{\sl \ott, \e\sr}}
\def\snt{\st}
\def\snti{\ensuremath{\sl \ott_i, \e_i\sr}}
\def\osnu{\ensuremath{\sl \otu, \f\sr}}
\def\osns{\ensuremath{\sl s, \g\sr}}
\def\snu{\su}
\def\fdist{\ensuremath{\d ist_\h}}
\def\dist{\ensuremath{Dist}}
\def\wtx{\ensuremath{\widetilde{X}}}

\def\bv{1{v}}
\def\bV{\mathbf{V}}
\def\bf{\mathbf{f}}
\def\bw{\mathbf{w}}
\def\zo{\mathbf{0}}
\def\bX{\mathbf{X}}
\def\bDelta{\mathbf{\Delta}}
\def\bdelta{\boldsymbol{\delta}}
\def\next{\mathcal{X}}
\def\until{\mathcal{U}}

\def\leqI{\ensuremath{\mathcal{SI}(\h)}}
\def\leqIq{\ensuremath{\mathcal{SI}_{\eqsim}(\h)}}
\def\oact{\ensuremath{\alpha}}
\def\oactb{\ensuremath{\beta}}
\def\sact{\ensuremath{\gamma}}
\def\fpi{\ensuremath{\widehat{\pi}}}
\newcommand{\supp}[1]{\ensuremath{\left\lceil{#1}\right\rceil}}
\newcommand{\support}[1]{\lceil{#1}\rceil}

\newcommand{\abis}{\stackrel{\lambda}\approx}
\newcommand{\abisa}[1]{\stackrel{#1}\approx}
\newcommand {\qbit} {\mbox{\bf{new}}}

\renewcommand{\theenumi}{(\arabic{enumi})}
\renewcommand{\labelenumi}{\theenumi}
\newcommand{\tr}{{\rm Tr}}
\newcommand{\rto}[1]{\stackrel{#1}\rightarrow}
\newcommand{\orto}[1]{\stackrel{#1}\longrightarrow}
\newcommand{\srto}[1]{\stackrel{#1}\longmapsto}
\newcommand{\sRto}[1]{\stackrel{#1}\Longmapsto}

\newcommand{\ass}[3]{\{#1\}\ #2\ \{#3\}}
\newcommand{\iass}[3]{[#1]\ #2\ [#3]}

\newcommand{\andor}{\ \&\ }

\newcommand {\true} {\ensuremath{{\mathbf{true}}}}
\newcommand {\false} {\ensuremath{{\mathbf{false}}}}
\newcommand {\abort}{\ensuremath{{\mathbf{abort}}}}
\newcommand {\sskip} {\mathbf{skip}}

\newcommand {\then} {\ensuremath{\mathbf{then}}}
\newcommand {\eelse} {\ensuremath{\mathbf{else}}}
\newcommand {\while} {\ensuremath{\mathbf{while}}}
\newcommand {\ddo} {\ensuremath{\mathbf{do}}}
\newcommand {\pend} {\ensuremath{\mathbf{end}}}
\newcommand {\inv} {\ensuremath{\mathbf{inv}}}

\newcommand {\mymeas} {\mathbf{meas}}

\newcommand \assert[1] {\mathbf{assert}\ #1}
\newcommand {\fail} {\mathbf{fail}}
\newcommand {\iif} {\mathbf{if}}
\newcommand {\fii} {\mathbf{fi}}
\newcommand {\od} {\mathbf{od}}
\def\mstm{\iif\ b\ \then\ S_1\ \eelse\ S_0\ \pend}
\def\wstm{\while\ b\ \ddo\ S\ \pend}

\newcommand\pmeasstm[3]{\iif\ #1\ \then\ #2\ \eelse\ #3\ \pend}
\def\pmstm{\iif\ P[\bar{q}] \ \then\ S_1\ \eelse\ S_0\ \pend}
\def\pwstm{\while\ P[\bar{q}] \ \ddo\ S\ \pend}

\newcommand\measstm[3]{\iif\ #1\ra\ #2\ \square\ \neg (#1)\ra #3\ \fii}

\newcommand\whilestm[2]{\while\ #1\ \ddo\ #2\ \pend}

\newcommand\alterex{\iif\ B_1\ra S_1 \square\ldots\square B_n\ra S_n\ \fii}

\newcommand\altercom{\iif\ \square_{i=1}^n B_i\ra S_i\ \fii}

\newcommand\repex{\ddo\ B_1\ra S_1 \square\ldots\square B_n\ra S_n\ \od}

\newcommand\repcom{\ddo\ \square_{i=1}^n B_i\ra S_i\ \od}

\newcommand\seqcom{\ddo\ \square_{i=1}^n B_i;\alpha_i\ra S_i\ \od}

\newcommand {\spann} {\mathrm{span}}

\newcommand{\rrto}[1]{\xhookrightarrow{#1}}
\newcommand{\con}[3]{\iif\ {#1}\ \then\ {#2}\ \eelse\ {#3}}

\newcommand{\Rto}[1]{\stackrel{#1}\Longrightarrow}
\newcommand{\nrto}[1]{\stackrel{#1}\nrightarrow}

\newcommand{\Rhto}[1]{\stackrel{\widehat{#1}}\Longrightarrow}
\newcommand{\define}{\ensuremath{\triangleq}}
\newcommand{\rsim}{\simeq}
\newcommand{\obis}{\approx_o}
\newcommand{\sbis}{\ \dot\approx\ } 
\newcommand{\stbis}{\ \dot\sim\ } 
\newcommand{\nssbis}{\ \dot\nsim\ } 

\newcommand{\bis}{\sim}
\newcommand{\rat}{\rightarrowtail}
\newcommand{\wbis}{\approx}
\newcommand{\id}{\mathcal{I}}
\newcommand{\stet}[1]{\{ {#1}  \}  } 
\newcommand{\unw}[1]{\stackrel{{#1}}\sim}
\newcommand{\rma}[1]{\stackrel{{#1}}\approx}

\def\step{\textsf{step}}
\def\obs{\textsf{obs}}
\def\dom{\textsf{dom}}
\def\purge{\textsf{ipurge}}
\def\source{\textsf{sources}}
\def\cnt{\textsf{cnt}}
\def\read{\textsf{read}}
\def\alter{\textsf{alter}}
\def\dirac#1{\delta_{#1}}

\def \srho {\sqrt{\rho}}
\def\tybool{\ensuremath{\mathbf{Boolean}}}
\def\tyint{\ensuremath{\mathbf{Integer}}}
\def\tyqubit{\ensuremath{\mathbf{Qubit}}}
\def\tyqudit{\ensuremath{\mathbf{Qudit}}}
\def\tyqureg{\ensuremath{\mathbf{Qureg}}}
\def\tyunitreg{\ensuremath{\mathbf{Unitreg}}}
\def\type{\ensuremath{\mathit{type}}}

\def\qconc{\mathcal{Q}}
\def\ps{\rm{PS}}
\def\is{\rm{IN}}
\def\aS{\sem{S}^{\#}}
\def\amap#1{\sem{#1}^{\#}}
\def\qstate{\rho}
\def\qassert{\Theta}
\def\qassertp{\Psi}
\def\casserts{\a}
\def\cstate{S}
\def\cstates{\prog}
\def\cassert{p}
\def\emptydis{\bot}
\def\qset{Q}
\def\qsetp{R}
\def\Exp{\mathrm{Exp}}

\def\supoprset{\mathbb{E}}

\def\leinf{\le_{\mathit{inf}}}
\def\geinf{\ge_{\mathit{inf}}}
\def\qstates{\s_V}
\def\qasserts{\a_V}
\def\qstatesh#1{\d(\mathcal{H}_{#1})}

\def\qassertsh#1{\mathcal{A}_{#1}}

\def\qstatesp{\mathcal{S}(\h')}

\newcommand\prog{\mathit{Prog}}
\def\ph{\ensuremath{\mathcal{P}(\h)}}
\def\phv{\ensuremath{\mathcal{P}(\h_V)}}

\def\l{\mathcal{L}}
\def\k{\mathcal{K}}
\def\qmc {\color{red}}
\def\dtmc {\color{black}}
\newcommand{\ysim}[1]{\stackrel{#1}\sim}
\def\z{\mathbf{0}}
\newcommand{\TRANDA}[3]{#1\xrightarrow{#2}_{{\sf D}}#3}
\def\pdist{\mathit{pDist}}

\def\C{\mathbb{C}}

\newcommand{\subs}[2]{{#2}/{#1}}

\def \Rm#1{\mbox{\rm #1}}
\def \lsem      {\raise1pt\hbox{\Rm {[\kern-.12em[}}}
\def \rsem      {\raise1pt\hbox{\Rm {]\kern-.12em]}}}
\def \sem#1{\mbox{\lsem$#1$\rsem}}
\def \asem#1{\mbox{\lsem$#1$\rsem$^\#$}}
\def \bsem#1{\mbox{\lsem$#1$\rsem$^b$}}

\def \asemp#1{\mbox{\lsem$#1$\rsem$_1^\#$}}
\def \asems#1{\mbox{\lsem$#1$\rsem$_\sigma^\#$}}
\def \asemsr#1{\mbox{\lsem$#1$\rsem$_{\sigma,\r}^\#$}}

\newtheorem{remark}{Remark}
	\def\pmstm{\iif\ P[\bar{q}] \ \then\ S_1\ \eelse\ S_0\ \pend}
	\def\pwstm{\while\ P[\bar{q}] \ \ddo\ S\ \pend}
	\def\qassert{a}
	\def\qassertp{b}
	\newsavebox{\tablebox}

\section{Introduction}

Abstract interpretation, originated by Cousot and Cousot~\cite{cousot1977abstract}, is a powerful technique for static analysis of program correctness. The key idea of abstract interpretation is to provide an over-approximation (abstraction) of the concrete program semantics. Consequently, analysis of programs can be done at the abstract level, which is usually much simpler, and correctness in the abstract domain implies correctness in the concrete domain. Over the past few decades, abstract interpretation has become increasingly popular in describing and analysing computational models in many different areas of computer science, such as model checking~\cite{clarke1994model,dams1997abstract,clarke2003counterexample,cousot2000temporal}, process calculi~\cite{cleaveland1994testing,venet1996abstract}, type inference~\cite{cousot1997types,comini2008polymorphic},  and theorem proving~\cite{plaisted1981theorem}.  More recently, analysis of quantum programs using abstract interpretation was proposed in~\cite{yu2021quantum}, where tuples of local subspaces of the state Hilbert space are regarded as abstraction of quantum states.

Hoare logic~\cite{hoare1969axiomatic} is one of the most popular syntax-oriented approaches for verifying the correctness of computer programs. 
The core notion of Hoare logic is the program correctness expressed in the form of Hoare triples $\ass{p}{S}{q}$ where $S$ is a program, and $p$ and $q$ are \emph{assertions} that describe the pre- and post-conditions of $S$, respectively. For non-probabilistic programs, assertions are typically first-order logic formulas. Intuitively, the triple $\ass{p}{S}{q}$ states that if $S$ is executed at a \emph{state} (evaluation of program variables) satisfying $p$ and it terminates, then $q$ must hold in the final state. This is called \emph{partial correctness}. If termination is further guaranteed in all states that satisfy $p$, then partial correctness becomes a \emph{total} one. 
Hoare logic provides a proof system which can systematically deduce the correctness of a program  represented by such a triple. After decades of development, Hoare logic has been successfully applied to the analysis of programs with non-determinism, recursion, parallel execution, probabilistic features, etc. For a detailed survey, please refer to~\cite{apt2019fifty,apt2010verification}.

In recent years, Hoare-type logics for quantum programs have been developed. Unlike the classical case, a logic system for assertions of quantum states was proposed only very recently~\cite{ying2022birkhoff}; most quantum Hoare logics developed so far simply take a certain semantic set as possible assertions. 
\cite{d2006quantum} proposes to regard positive operators not greater than  (w.r.t. L\"{o}wner order) the identity operator as (quantitative)  assertions of quantum states. Then the \emph{degree} of a quantum state $\rho$ satisfying an assertion $M$ is denoted $\tr(M\rho)$, the expected value of outcomes if $\rho$ is measured according to the projective measurement determined by $M$.  A Hoare triple for quantum programs then has the form $\ass{M}{S}{N}$ where $S$ is a quantum program, and $M$ and $N$ are quantum assertions, and it is partially (resp. totally) correct if $$\tr(M\rho) \leq \tr(N\cdot\sem{S}(\rho)) + \tr(\rho) - \tr(\sem{S}(\rho))$$
 (resp. $\tr(M\rho) \leq \tr(N\cdot \sem{S}(\rho))$) for any quantum state $\rho$~\cite{ying2012floyd,ying2022proof,ying2016foundations,feng2022verification,feng2021quantum}. Note that the term $\tr(\rho) - \tr(\sem{S}(\rho))$ appearing in partial correctness but not in total correctness denotes the probability for $S$ to diverge (not terminate) at $\rho$. Another line of research, which is conceptually and computationally simpler, is to regard subspaces (or equivalently, orthogonal projectors) of the associated Hilbert space as (qualitative) assertions, and a quantum state $\rho$ satisfies a subspace assertion $P$ iff the support (the image space of linear operators) of $\rho$ is included in $P$~\cite{zhou2019applied,unruh2019quantum}. Partial correctness of $\ass{P}{S}{Q}$ means that $\sem{S}(\rho)$ satisfies $Q$ as long as $\rho$ satisfies $P$, similar to the classical case. Total correctness further requires that $\tr(\sem{S}(\rho)) = \tr(\rho)$ for all $\rho$ satisfying $P$. Obviously, subspace based Hoare logics are special cases of positive operator based ones, by noting that projectors are positive operators with eigenvalues being either 0 or 1.
 For comparison with abstract interpretation, we will consider this simplified form of quantum Hoare logic in this paper. 

Incorrectness logic~\cite{o2019incorrectness}, or reverse Hoare logic~\cite{vries2011reverse}, is a complementary method to reason about the \emph{incorrectness} of programs. Similar to Hoare logic, the key notion of incorrectness logic is a triple $\iass{p}{S}{q}$ which asserts that any state satisfying $q$ is reachable from a state satisfying $p$ by executing the program $S$. Note that the postcondition $q$ in the Hoare triple $\ass{p}{S}{q}$ provides an \emph{over-approximation} of the set of final states when starting with states in $p$, while $q$ in the incorrectness triple $\iass{p}{S}{q}$ provides an \emph{under-approximation} of the same set. Again, incorrectness logic has recently been extended to analyse quantum programs where quantum assertions are taken as subspaces of the associated Hilbert space~\cite{yan2022incorrectness}.

So far, the aforementioned approaches for analysis of quantum programs, namely abstract interpretation, Hoare logic, and incorrectness logic, have been developed largely in parallel. In this paper we analyse the relationship between them. Our discovery is twofold:

\begin{enumerate}
	\item Given a quantum abstract interpretation in which the abstract domain for quantum states is well-structured, a quantum Hoare logic is naturally induced which is sound  (resp. sound and relatively complete) if the abstract operator is sound (resp. complete) for each basic command of the quantum language under consideration. Similar results apply to quantum incorrectness logic as well. Compared to the applied Hoare logic~\cite{zhou2019applied} and incorrectness logic~\cite{yan2022incorrectness} for quantum programs, our induced logic systems are in a forward fashion, making them more useful in certain applications.
	
	\item Conversely, for any quantum Hoare logic in which the set of assertions for quantum states is well-structured, a quantum abstract interpretation is naturally induced which is sound (resp. complete) if the Hoare logic is sound (resp. sound and relatively complete). Again, similar results apply to quantum incorrectness logic as well. As an application, these results imply the non-existence of any sound and relatively complete Hoare or incorrectness logic for quantum programs if tuples of local subspaces are taken as assertions.  
\end{enumerate}

The rest of this paper is organised as follows. We review in Sec.~\ref{sec:pre} some basic notions from abstract interpretation and quantum computing that will be used throughout this paper. In Sec.~\ref{sec:qlang}, a simple quantum while-language is introduced, which serves as the target language of our analysis, and its concrete denotational semantics is defined. We examine the relationship between well-structured abstract domains and sets of assertions for quantum states in Sec.~\ref{sec:advsqa}, which sets the stage for the discussion that follows. Sec.~\ref{sec:hlvsai} is the main part of this paper, where we elaborate on how a sound (resp. sound and relatively complete) Hoare logic and a sound (resp. complete) abstract interpretation of quantum programs can be derived from each other. Similar results are also discussed for incorrectness logic. Finally, Sec.~\ref{sec:con} concludes the paper and points out some directions for future study.

\section{Preliminaries}
\label{sec:pre}
This section is devoted to fixing some notations from abstract interpretation and quantum computing that will be used in this paper. For a thorough introduction to the relevant backgrounds, please refer to~\cite{mine2017tutorial} (abstract interpretation) and~\cite{nielsen2002quantum} (quantum computing).

	\subsection{Abstract Interpretation}

Let the concrete domain for program states be a partially ordered set, a.k.a poset, $(C, \leq_C)$. Typically, elements in $C$ are subsets of program states, and $\leq_C$ is just the set inclusion. Let the abstract domain be another poset $(A, \leq_A)$. The concrete and abstract domains are related by a pair of monotonic functions $\alpha: C\ra A$ and $\gamma: A\ra C$. The pair $(\alpha, \gamma)$ is said to form a \emph{Galois connection}, denoted $(C, \leq_C) \galois{\alpha}{\gamma} (A, \leq_A)$, if for all $c\in C$ and $a\in A$, $c\leq_C \gamma(a)$ iff $\alpha(c) \leq_A a$.  Furthermore, if $\alpha\circ \gamma = \mathrm{id}_A$ then $(\alpha, \gamma)$ forms a \emph{Galois embedding}, where $\mathrm{id}_A$ is the identity relation over $A$. Note that a Galois connection $(C, \leq_C) \galois{\alpha}{\gamma} (A, \leq_A)$ is a Galois embedding iff any of the following holds: (1) $\alpha$ is surjective, that is, for all $a\in A$, there exists $c\in C$ such that $\alpha(c) = a$; (2) $\gamma$ is injective, that is, for all $a,a'\in A$, $\gamma(a) = \gamma(a')$ implies $a=a'$. 

Given an operator $f: C\ra C$ in the concrete domain, an abstract operator $f^\#: A\ra A$ is called a \emph{sound abstraction} of $f$ if 
$\alpha\circ f \leq_A f^\#\circ\alpha$, and it is \emph{complete} if $\alpha\circ f = f^\#\circ\alpha$. It is easy to check that complete abstractions are closed under composition; that is, if $f^\#$ and $g^\#$ are complete abstractions of $f$ and $g$ respectively, then $f^\#\circ g^\#$ is a complete abstraction of $f\circ g$.
Note that a complete abstraction does not necessarily exist. However,
the \emph{best abstraction} of $f$, defined as $\alpha\circ f\circ \gamma$, always exists. It is the smallest one among all sound abstractions.

\begin{remark}
	In the literature, there are alternative definitions of sound and complete abstraction using the concretisation function $\gamma$ instead of the abstraction function $\alpha$. Specifically, an abstract operator $f^\#$ of $f$ is sound if $f\circ \gamma \leq_C \gamma \circ f^\#$, and it is complete if $f\circ \gamma = \gamma \circ f^\#$. It can be easily checked that when $(\alpha, \gamma)$ forms a Galois connection, these two notions of soundness are equivalent; that is, $\alpha\circ f \leq_A f^\#\circ\alpha$ iff 
	$f\circ \gamma \leq_C \gamma \circ f^\#$. However, the two notions of completeness are in general incomparable. Nevertheless, in either case the complete abstraction, if exists, must be the best abstraction.
		
\end{remark}

\subsection{Basic quantum computing}

Let $\h$ be a finite-dimensional Hilbert space, and $\dim(\h)$ denote its dimension. Following the tradition of quantum computing, vectors in $\h$ are denoted in the Dirac form $|\psi\>$. The inner and outer products of two vectors $|\psi\>$ and $|\phi\>$ are written as $\<\psi|\phi\>$ and $|\psi\>\<\phi|$ respectively. Let $\lh$ be the set of linear operators on $\h$ and $A\in \lh$. Denote by $\tr(A)= \sum_{i\in I} \<\psi_i|A|\psi_i\>$ the \emph{trace} of $A$ where $\{|\psi_i\> : i\in I\}$ is an orthonormal basis of $\h$.  The \emph{adjoint} of $A$, denoted $A^\dag$, is the unique linear operator in $\lh$ such that $\<\psi|A|\phi\> = \<\phi|A^\dag |\psi\>^*$ for all $|\psi\>, |\phi\>\in \h$. Here, for a complex number $z$, $z^*$ denotes its conjugate. An operator $A\in \lh$ is said to be (1) \emph{hermitian} if $A^\dag = A$; (2) \emph{unitary} if $A^\dag A = I_\h$, the identity operator on $\h$; (3) \emph{positive} if for all $|\psi\>\in \h$, $\<\psi|A|\psi\>\geq 0$. 
Every hermitian operator $A$ has a \emph{spectral decomposition} form $A  = \sum_{i\in I} \lambda_i |\psi_i\>\<\psi_i|$ where $\{|\psi_i\> : i\in I\}$ constitute an orthonormal basis of $\h$. 
The L\"owner (partial) order $\le$ on $\lh$ is defined by letting $A\le B$ iff $B-A$ is positive. 

A linear operator $\e$ from $\l(\h_1)$ to $\l(\h_2)$ is called a \emph{super-operator}.  It is said to be (1) \emph{positive} if it maps positive operators on $\h_1$ to positive operators on $\h_2$; (2) \emph{completely positive} if $\mathcal{I}_\h\otimes \e$ is positive for all finite dimensional Hilbert space $\h$, where $\mathcal{I}_\h$ is the identity super-operator on $\lh$; (3) \emph{trace-preserving} (resp. \emph{trace-nonincreasing}) if 
$\tr(\e(A)) = \tr(A)$ (resp. $\tr(\e(A)) \leq \tr(A)$ for any positive operator $A\in \l(\h_1)$. Given the tensor product space $\h_1\otimes \h_2$, the \emph{partial trace} with respect to $\h_2$, denoted $\tr_{\h_2}$, is a linear mapping from $\l(\h_1\otimes \h_2)$ to $\l(\h_1)$ such that for any $|\psi_i\>, |\phi_i\> \in \h_i$, 
$$\tr_{\h_2}(|\psi_1\>\<\phi_1|\otimes |\phi_1\>\<\phi_2|) = 
\<\phi_2|\phi_1\> |\psi_1\>\<\phi_1|.$$
The definition extends linearly to $\l(\h_1\otimes \h_2)$.

According to von Neumann's formalism of quantum mechanics
\cite{vN55}, any quantum system with finite degrees of freedom is associated with a finite-dimensional Hilbert space $\h$ called its \emph{state space}. When $\dim(\h) = 2$, such a system is called a \emph{qubit}, the analogy of bit in classical computing. A {\it pure state} of the system is described by a normalised vector in $\h$. When the system is in state $|\psi_i\>$ with probability $p_i$, $i\in I$, it is in a \emph{mixed} state, represented by the \emph{density operator} $\sum_{i\in I} p_i|\psi_i\>\<\psi_i|$ on $\h$. Obviously, a density operator is positive and has trace 1. 
In this paper, we follow Selinger's convention~\cite{selinger2004towards} to  regard \emph{partial density operators}, i.e. positive operators with traces not greater than 1 as (unnormalised) quantum states. Intuitively, a partial density operator $\rho$ denotes a legitimate quantum state $\rho/\tr(\rho)$ which is obtained with probability $\tr(\rho)$. Denote by $\dh$ the set of partial density operators on $\h$. The state space of a composite system (e.g., a quantum system consisting of multiple qubits) is the tensor product of the state spaces
of its components. For any $\rho$ in $\d(\h_1 \otimes \h_2)$, the partial traces
$\tr_{\h_1}(\rho)$ and $\tr_{\h_2}(\rho)$ are
the reduced quantum states of $\rho$ on $\h_2$ and $\h_1$, respectively. 

The \emph{evolution} of a closed quantum system is described by a unitary
operator on its state space: if the states of the system at 
$t_1$ and $t_2$ are $\rho_1$ and $\rho_2$, respectively, then
$\rho_2=U\rho_1U^{\dag}$ for some unitary $U$. The general dynamics that can occur in a physical system is
described by a completely positive and trace-preserving super-operator. Note that the unitary transformation $\e_U(\rho)\define U\rho U^\dag$ is such a super-operator. 
A (projective) quantum {\it measurement} $\m$ is described by a
collection $\{P_i : i\in O\}$ of projectors (hermitian operators with eigenvalues being either 0 or 1) in the state space $\h$, where $O$ is the set of measurement outcomes. It is required that the
measurement operators $P_i$'s satisfy the completeness equation
$\sum_{i\in I}P_i = I_\h$. If the system was in state $\rho$ before measurement, then the probability of observing outcome $i$ is given by
$p_i=\tr(P_i\rho),$ and the state of the post-measurement system
becomes $\rho_i = P_i\rho P_i/p_i$ whenever $p_i>0$. Sometime we use 
a hermitian operator $M$ in $\lh$ called \emph{observable} to represent a  projective measurement. To be specific, let 
\[
M=\sum_{m\in \mathit{spec}(M)}mP_m
\] 
where $\mathit{spec}(M)$ is the set of eigenvalues of $M$, and $P_m$ the projection onto the eigenspace associated with $m$. Then the projective measurement determined by $M$ is $\{P_m : m\in \mathit{spec}(M)\}$.

	\section{A simple quantum while-language}
	\label{sec:qlang}
	
	The target quantum language of our analysis is
	an extension of the purely quantum while-language defined in~\cite{feng2007proof,ying2012floyd} with assertions. 
	Let $\QVar$, ranged over by $q, r, \cdots$, be a finite set of (qubit-type) quantum variables. For any subset $W$ of $\QVar$, let
	\[\h_W \define \bigotimes_{q\in W} \h_{q},
	\]
	where $\h_{q}$ is the 2-dimensional Hilbert space associated with $q$. As we use subscripts to distinguish Hilbert spaces with different quantum variables, their order in the tensor product is irrelevant. 
	
	The syntax of our language is given as follows:
	\begin{equation*}
		\begin{array}{cclr}
			S\quad  & ::= & \sskip & \textit{(no-op)}\\
			&|&  \bar{q}:=|0\> & \textit{(initialisation)}\\
			&| & 	\bar{q}\apply U &\textit{(unitary operation)}\\
			&|&   \assert{P[\bar{q}] }& \textit{(assertion)}\\
			&|&   S_0;S_1&\textit{(sequence)}\\
			&| &  \pmstm & \textit{(conditional)}\\
			&|&  \pwstm	& \textit{(loop)}		
		\end{array}
	\end{equation*}
	where $S,S_0$ and $S_1$ are quantum programs, $\bar{q} \define q_1, \ldots, q_t$ a (ordered) tuple of distinct quantum variables from $V$, $U$ a unitary operator on
	$\h_{\bar{q}}$, and $P$ a subspace of $\h_{\bar{q}}$.
	Sometimes we also use $\bar{q}$ to denote the (unordered) set $\{q_1,q_2,\dots,q_t\}$.

	For clarification, we often use subscripts to emphasise the quantum system on which an operator is performed. For example, $P_W$ means $P$ acting on system $W$. To simplify notations, we do not distinguish $P_W$ with its cylindrical extension $P_W\otimes I_{V\backslash W}$ to $\h_V$. Furthermore, we use the same symbol, say $P$, to denote both a subspace  and its corresponding projector. The correct meaning of theses notations should be clear from the context. Consequently, a quantum state $|\psi\> \in P$ iff $P|\psi\> = |\psi\>$. Here, the former $P$ denotes a subspace, while the latter one denotes the corresponding projector.

	\begin{definition}[Denotational semantics]\label{def:denotational}
		Let $S$ be a  quantum program. The \emph{denotational semantics} of $S$ is a mapping 
		$\sem{S} : \dhv\ra  \dhv$ defined inductively in Fig.~\ref{fig:semantics}, where $P_{\bar{q}}^\perp = I_{\h_V} - P_{\bar{q}}$ is the projector onto the orthocomplement of $P_{\bar{q}}$ in $\h_V$. 
	\end{definition}

	Intuitively, the $\sskip$ statement does not change the input state, while $\bar{q} := |0\>$ sets the system $\bar{q}$ to state $|0\>$ where $|\bar{q}|= t$. Note that $|i\>_{\bar{q}}\<j|$ denotes the operator $|i\>\<j|$ acting on $\bar{q}$, and $\{|0\>, \cdots, |2^t-1\>\}$ constitute the computational basis of $\h_{\bar{q}}$. The statement $\bar{q}\apply U$ applies the unitary operator $U$ on $\bar{q}$, while $\assert{P[\bar{q}]}$ measures system $\bar{q}$ according to the projective measurement $\{P, P^\perp\}$ and post-selects the outcome for $P$; that is, if $P^\perp$ is observed, then the program aborts without outputting anything (or equivalently, it outputs the zero operator). Note also that here we adopt the convenience 
	of using a partial density operator $\rho$ to encode both the (normalised) quantum state $\rho/\tr(\rho)$ and the probability $\tr(\rho)$ of reaching it~\cite{selinger2004towards}. Similarly, the branching statement $\pmstm$ also measures system $\bar{q}$ according to the projective measurement $\{P, P^\perp\}$ and then executes $S_1$ or $S_0$ depending on the measurement outcome. The output of this statement is defined as the combination of the two branches. Again, thanks to the convention of partial density operators, this combination is simply the summation of the output states from both branches. Finally, the while loop $\pwstm$ takes into account the output states from different iterations. The well-definedness of the semantics of while loops comes from the fact that the set $\d(\h_V)$ of partial density operators on $\h_V$ is a complete partial order set~\cite{selinger2004towards,ying2016foundations}.

	\begin{figure}[t]
		\begin{align*}
			\sem{\sskip}(\rho) &= \rho\\
			\sem{\bar{q}:=|0\>}(\rho) &= \sum_{i=0}^{2^t-1} |0\>_{\bar{q}}\<i|\rho|i\>_{\bar{q}}\<0|\\
			\sem{\bar{q}\apply U}(\rho) &=U_{\bar{q}}\rho U^\dag_{\bar{q}}\\
			\sem{\assert{P[\bar{q}]}}(\rho) &= P_{\bar{q}}\rho P_{\bar{q}}\\		
			\sem{S_0; S_1}(\rho) &= \sem{S_1}\circ \sem{S_0}(\rho)\\
			\sem{\pmstm}(\rho) &=  \sem{\assert P[\bar{q}]; S_1}(\rho)+\sem{\assert{P^\perp[{\bar{q}}]}; S_0} (\rho) \\
			\sem{\pwstm}(\rho) &= \sum_{i=0}^\infty \sem{(\assert P[\bar{q}]; S)^i; \assert{P^\bot[\bar{q}]}}(\rho)
		\end{align*}
		\caption{Denotational semantics for  quantum programs. 
		}
		\label{fig:semantics}
	\end{figure}

	For the purpose of abstract interpretation for quantum programs, we take, and fix throughout this paper, 
	\[
		(\qconc \define 2^{\dhv}, \subseteq, \cup, \cap, \emptyset, \dhv)
	\]
	to be the concrete domain for their (collecting) semantics, 
	where $\cup$ and $\cap$ are respectively the normal union and intersection over sets of quantum states. Thus $\qconc$ is a complete lattice. The denotational semantics defined in Definition~\ref{def:denotational} is then extended to the concrete domain by letting $\sem{S}(R) = \{\sem{S}(\rho) : \rho\in R\}$ for any $R\in \qconc$. Obviously, such defined $\sem{S}$ is monotonic with respect to $\subseteq$; that is, $\sem{S}(R)\subseteq \sem{S}(R')$ whenever $R\subseteq R'$.

\section{Abstract quantum domains v.s. quantum assertions}
\label{sec:advsqa}
The main contribution of this paper is a close relationship between abstract interpretation and Hoare/incorrectness logic for quantum programs. To this end, we first  examine the relationship between abstract domains and assertions for quantum states. For the sake of simplicity, we assume that both the abstract quantum domain and the quantum assertion set are taken as a complete lattice. Furthermore, note that the concrete domain $\qconc$ of quantum states defined in Sec.~\ref{sec:qlang} enjoys a linear structure: for  any $\rho_1, \rho_2\in \dhv$ and $x_1, x_2 >0$, it holds that $x_1\rho_1 +x_2\rho_2 \in \dhv$ whenever $\tr(x_1\rho_1 +x_2\rho_2) \leq 1$. To respect this structure, we put some natural restrictions on both the abstract domain and the assertion set. 

\subsection{Well-structured abstract quantum domain}
\label{sec:wsad}

Let $(\a, \leq_\a, \vee, \wedge, \bot, \top)$ be a complete lattice. For $\a$ to be an abstract domain for the set of quantum state, we assume a
pair of monotonic functions $\alpha: \qconc\ra \a$ (abstraction) and $\gamma: \a\ra \qconc$ (concretisation). 

\begin{definition}\label{def:wsdomain}
	The complete lattice $\a$ as an abstract domain of $\qconc$ is said to be \emph{well-structured}, if 
	\begin{enumerate}
		\item $(\alpha, \gamma)$ forms a Galois embedding between $\qconc$ and $\a$; and 
		\item	for any $\rho_i\in \dhv$ and $x_i >0$, $i=1,2,\ldots$, with $\sum_{i}x_i\rho_i \in \dhv$,
\begin{equation}\label{eq:wels}
	\alpha\left(\sum_{i}x_i\rho_i\right) = \bigvee_i \alpha(\rho_i).
\end{equation}
Here and in the following, we abbreviate $\alpha(\{\rho\})$ into $\alpha(\rho)$ for simplicity.
\end{enumerate}
\end{definition}

Note that if  $(\alpha, \gamma)$ forms a Galois connection, then 	$\alpha\left(\bigcup_{i}\ \{\rho_i\}\right) = \bigvee_i \alpha(\rho_i)$. Thus the second condition in the above definition can be regarded as an analogy of this property for linear combination of quantum states.

The following lemma gives equivalent characterisations of the second condition in Definition~\ref{def:wsdomain} in terms of the concretisation function.

\begin{lemma}\label{lem:spbased}
	Let $(\a, \leq_\a, \vee, \wedge, \bot, \top)$ be a complete lattice, with $\alpha: \qconc\ra \a$ and $\gamma: \a\ra \qconc$ forming a Galois embedding. Then the following three statements are equivalent:
	\begin{enumerate}
		\item $\a$ is well-structured;
		\item for any $a\in \a$, $\rho_i\in \dhv$, and $x_i >0$ with $\sum_{i}x_i\rho_i \in \dhv$,
		\begin{equation}\label{eq:well-con}
			\sum_{i}x_i\rho_i\in \gamma(a)\quad \Leftrightarrow\quad \forall i, \rho_i\in \gamma(a);
		\end{equation}
		\item for any $a\in \a$, $\gamma(a)$ as a subset of $\d(\h_V)$ is 
		\begin{itemize}
			\item convex;
			\item $\omega$-cpo: if for each $i\geq 1$, $\rho_i \in \gamma(a)$ and $\rho_i \le \rho_{i+1}$, then $\bigsqcup_{i} \rho_i \in \gamma(a)$;
			\item down-closed: if $\rho \le \sigma$ and $\sigma\in \gamma(a)$, then $\rho\in \gamma(a)$ as well; and
			\item closed under positive scalar multiplication: if $\rho \in \gamma(a)$ and $x\rho\in \d(\h_V)$ with $x>0$, then $x\rho\in \gamma(a)$ as well. 
		\end{itemize} 
	\end{enumerate}
\end{lemma}
\begin{proof}
	 (1) $\Leftrightarrow$ (2): Direct from the observation that 
	\begin{align*}
		\sum_{i}x_i\rho_i\in \gamma(a)&\quad \Leftrightarrow\quad 	\alpha\left(\sum_{i}x_i\rho_i\right) \leq_\a a
	\end{align*}
	while
	\[
	\forall i, \rho_i\in \gamma(a) \quad \Leftrightarrow	\quad \forall i,\  \alpha(\rho_i)\leq_\a a \quad  \Leftrightarrow \quad \bigvee_i \alpha(\rho_i)\leq_\a a \]
	
	(2) $\Rightarrow$ (3): For any $a\in \a$, it is easy to see from the sufficiency part of Eq.\eqref{eq:well-con} that $\gamma(a)$ is convex and closed under positive scalar multiplication. Now if $\rho\le \sigma$ and $\sigma\in \gamma(a)$, then $\sigma - \rho \in \d(\h_V)$ as well. From the fact that $\rho + (\sigma - \rho) = \sigma$, we know $\rho\in \gamma(a)$ from the necessity part of Eq.\eqref{eq:well-con}. Thus $\gamma(a)$ is down-closed. 
	
	Finally, if for each $i\geq 1$, $\rho_i \in \gamma(a)$ and $\rho_i \le \rho_{i+1}$, then $\rho_{i+1} - \rho_i \in \gamma(a)$ as well. From the fact that
	\[
	\bigsqcup_{i} \rho_i = \rho_1 + \sum_{i\geq 1} \left(\rho_{i+1} - \rho_i\right),
	\]
	we know $\bigsqcup_{i} \rho_i\in \gamma(a)$ from the sufficiency part of Eq.\eqref{eq:well-con}. Thus $\gamma(a)$ is an $\omega$-cpo. 
	
	(3) $\Rightarrow$ (2): The proof consists of two parts:
	\begin{align*}
		\sum_{i}x_i\rho_i\in \gamma(a)\quad & \Rightarrow \quad \forall i, x_i\rho_i\in \gamma(a) \quad\qquad \textrm{(down-closed)}\\
		& \Rightarrow \quad \forall i, \rho_i\in \gamma(a) \qquad\qquad \textrm{(scaling)}
	\end{align*}
and
	\begin{align*}
	\forall i, \rho_i\in \gamma(a)\quad & \Rightarrow \quad \forall n>0,\ \sum_{i=1}^n \frac{x_i}{\sum_{i=1}^nx_i}\rho_i\in \gamma(a) \quad\qquad \textrm{(convexity)}\\
	& \Rightarrow \quad \forall n>0,\ \sum_{i=1}^n x_i\rho_i\in \gamma(a) \hspace{5.3em} \textrm{(scaling)}\\
	& \Rightarrow \quad \sum_{i}x_i\rho_i\in \gamma(a)\hspace{8.7em}  \textrm{($\omega$-cpo)} \qedhere
\end{align*}
\end{proof}

\begin{example}[Subspace abstract domain]\label{exm:allspaces}
	The simplest example of well-structured abstract domain for quantum states is the subspace domain
	\[
		(\shv, \subseteq, \vee, \cap, \{0\}, \h_V)
	\]
	where $\shv$ is the set of all subspaces (or equivalently, all projectors) of $\h_V$, $\subseteq$ is the ordinary subset relation between subspaces (or equivalently, the L\"{o}wner order between the corresponding projectors), and the join $P \vee Q = \spann(P\cup Q)$ is defined as the subspace spanned by elements from $P$ or $Q$.  Then this domain is a complete lattice.
	Let the concretisation and abstraction functions $\gamma_s : \shv \ra \qconc$ and $\alpha_s : \qconc \ra \shv$ be defined as follows: for any $P\in \shv$ and $R\in \qconc$,
	\begin{align*}
		\gamma_s(P) &\define \left\{\rho\in \dhv: \supp{\rho}\subseteq P\right\},\\
		\alpha_s(R) &\define \bigvee\left\{\supp{\rho} : \rho \in R\right\}.
	\end{align*}
	Here $\supp{\rho}$ denotes the support subspace of $\rho$. We now show that such defined $\alpha_s$ and $\gamma_s$ form a Galois embedding between the concrete domain $\qconc$ and the abstract domain $\shv$:
		\begin{align*}
			R \subseteq \gamma_s(P)\quad &
			\Leftrightarrow\quad  R \subseteq \left\{\rho\in \dhv: \supp{\rho}\subseteq P\right\}\\
			&
			\Leftrightarrow\quad  \forall \rho\in R,  \supp{\rho}\subseteq P\\
			&
			\Leftrightarrow\quad  \alpha_s(R) \subseteq P	
		\end{align*}
		where the necessity part of the last equivalence is from the fact that $P$ is a subspace. Furthermore, it is easy to check that $\alpha_s$ is surjective: for any $P\in \shv$, $\alpha_s(\gamma_s(P)) = P$.

		Finally, for any $\rho_i\in \dhv$ and $x_i >0$, $i=1,2,\ldots$, with $\sum_{i}x_i\rho_i \in \dhv$,
		\[
			\alpha_s\left(\sum_{i}x_i\rho_i\right) = \supp{\sum_{i}x_i\rho_i} =  \bigvee_i \supp{\rho_i}=  \bigvee_i \alpha_s(\rho_i).
		\]
		Thus $\shv$ is a well-structured abstract domain for quantum states in $\dhv$. 
\end{example}

We now show that the subspace abstract domain $\shv$ presented in Example~\ref{exm:allspaces} is actually the \emph{most concrete} well-structured abstract quantum domain, in the sense that any other well-structured abstract domain can be regarded as an abstraction of $\shv$ in terms of a Galois embedding. 
For this purpose we first prove 
\begin{lemma}\label{lem:tmp}
	Let $\a$ be a well-structured abstract domain for quantum states, with $\alpha: \qconc\ra \a$ being the abstraction function. For any $R_1, R_2\subseteq \dhv$, 
	$$\alpha_s(R_1) = \alpha_s(R_2)\quad \Rightarrow \quad\alpha(R_1) = \alpha(R_2).$$
	Consequently, we have $\alpha = \alpha\circ \gamma_s\circ \alpha_s$.
\end{lemma}
\begin{proof} 
	Suppose $\alpha_s(R_1) = \alpha_s(R_2)$.
	For any $\rho\in R_1$, as $\h_V$ is finite dimensional, we can always find a set of states $\rho_1, \ldots, \rho_n$ in $R_2$, $\tr(\rho_i)=1$,  and $x>0$ such that $x\rho \le  \sum_{i=1} \rho_i/n$. Thus from the assumption that $\a$ is a well-structured, we have
	$$\alpha(\rho) = \alpha(x\rho) \leq_{\a} \bigvee_{i=1}^n\alpha(\rho_i) \leq_{\a} \alpha(R_2).
	$$
	Thus $\alpha(R_1) = \textstyle\bigvee\{\alpha(\rho): \rho\in R_1\}\leq_{\a}\alpha(R_2).$ The other direction can be similarly proved.
	
	The last part of the lemma follows from the observation that for any $R\subseteq \dhv$, $\alpha_s(R) =\alpha_s\circ \gamma_s\circ\alpha_s(R)$. 
\end{proof}
\begin{theorem}
	Let $\a$ be a well-structured abstract domain for quantum states, with $\alpha: \qconc\ra \a$ and $\gamma: \a\ra \qconc$ being the abstraction and concretisation functions respectively. Then there exists a Galois embedding $(\alpha',\gamma')$ with $\alpha': \shv\ra \a$  and $\gamma': \a\ra \shv$ such that $\alpha=\alpha'\circ \alpha_s$ and $\gamma = \gamma_s\circ\gamma'$. That is, the following diagram commutes for $\alpha$'s and $\gamma$'s respectively:
\begin{center}

\tikzset{every picture/.style={line width=0.75pt}} 

\begin{tikzpicture}[x=0.75pt,y=0.75pt,yscale=-1,xscale=1]
	
	\draw    (252.5,53) -- (344,53.98) ;
	\draw [shift={(346,54)}, rotate = 180.61] [color={rgb, 255:red, 0; green, 0; blue, 0 }  ][line width=0.75]    (10.93,-3.29) .. controls (6.95,-1.4) and (3.31,-0.3) .. (0,0) .. controls (3.31,0.3) and (6.95,1.4) .. (10.93,3.29)   ;
	\draw    (255,48) -- (347,48) ;
	\draw [shift={(253,48)}, rotate = 0] [color={rgb, 255:red, 0; green, 0; blue, 0 }  ][line width=0.75]    (10.93,-3.29) .. controls (6.95,-1.4) and (3.31,-0.3) .. (0,0) .. controls (3.31,0.3) and (6.95,1.4) .. (10.93,3.29)   ;
	\draw    (237.5,68) -- (237.98,123) ;
	\draw [shift={(238,125)}, rotate = 269.5] [color={rgb, 255:red, 0; green, 0; blue, 0 }  ][line width=0.75]    (10.93,-3.29) .. controls (6.95,-1.4) and (3.31,-0.3) .. (0,0) .. controls (3.31,0.3) and (6.95,1.4) .. (10.93,3.29)   ;
	\draw    (231.03,71) -- (232,127) ;
	\draw [shift={(231,69)}, rotate = 89.01] [color={rgb, 255:red, 0; green, 0; blue, 0 }  ][line width=0.75]    (10.93,-3.29) .. controls (6.95,-1.4) and (3.31,-0.3) .. (0,0) .. controls (3.31,0.3) and (6.95,1.4) .. (10.93,3.29)   ;
	\draw    (339,67) -- (270.51,126.69) ;
	\draw [shift={(269,128)}, rotate = 318.93] [color={rgb, 255:red, 0; green, 0; blue, 0 }  ][line width=0.75]    (10.93,-3.29) .. controls (6.95,-1.4) and (3.31,-0.3) .. (0,0) .. controls (3.31,0.3) and (6.95,1.4) .. (10.93,3.29)   ;
	\draw    (340.49,73.32) -- (271,134) ;
	\draw [shift={(342,72)}, rotate = 138.87] [color={rgb, 255:red, 0; green, 0; blue, 0 }  ][line width=0.75]    (10.93,-3.29) .. controls (6.95,-1.4) and (3.31,-0.3) .. (0,0) .. controls (3.31,0.3) and (6.95,1.4) .. (10.93,3.29)   ;
	
	\draw (356,46) node [anchor=north west][inner sep=0.75pt]   [align=left] {$\displaystyle \mathcal{A}$};
	\draw (290,56) node [anchor=north west][inner sep=0.75pt]   [align=left] {$\displaystyle \alpha $};
	\draw (291,31) node [anchor=north west][inner sep=0.75pt]   [align=left] {$\displaystyle \gamma $};
	\draw (218,134.4) node [anchor=north west][inner sep=0.75pt]    {$\mathcal{S}(\mathcal{H}_{V})$};
	\draw (243,89) node [anchor=north west][inner sep=0.75pt]   [align=left] {$\displaystyle \alpha _{s}$};
	\draw (214,88) node [anchor=north west][inner sep=0.75pt]   [align=left] {$\displaystyle \gamma _{s}$};
	\draw (300.5,102.5) node [anchor=north west][inner sep=0.75pt]   [align=left] {$\displaystyle \alpha '$};
	\draw (286,86) node [anchor=north west][inner sep=0.75pt]   [align=left] {$\displaystyle \gamma '$};
	\draw (230,44) node [anchor=north west][inner sep=0.75pt]   [align=left] {$\mathcal{Q}$};

\end{tikzpicture}

\end{center}
\end{theorem}
\begin{proof}
Let $\alpha' \define \alpha\circ \gamma_s$ and $\gamma' \define \alpha_s\circ \gamma$. From the fact that $\gamma\circ \alpha \geq _{\qconc} \mathrm{id}_{\qconc}$ we have 
$$\gamma'\circ \alpha' =  \alpha_s\circ \gamma\circ \alpha\circ \gamma_s \geq_{\shv} \alpha_s\circ \gamma_s = \mathrm{id}_{\shv}.$$
Furthermore, from Lemma~\ref{lem:tmp}, we have $\alpha'\circ \alpha_s = \alpha\circ \gamma_s\circ \alpha_s = \alpha$ and so
$$\alpha'\circ \gamma' =  \alpha'\circ\alpha_s\circ \gamma =  \alpha\circ \gamma = \mathrm{id}_{\a}.$$
Thus $\shv\galois{\alpha'}{\gamma'} \a$ is indeed a Galois embedding. Finally, from $\gamma_s\circ \alpha_s \geq _{\qconc} \mathrm{id}_{\qconc}$ we have
$$\gamma_s\circ\alpha_s\circ\gamma \geq_{\a} \gamma.$$
On the other hand, 
$$\gamma = \gamma\circ\alpha\circ\gamma = \gamma\circ\alpha\circ\gamma_s\circ\alpha_s\circ\gamma\geq_{\a}\gamma_s\circ\alpha_s\circ\gamma.$$
Thus $\gamma_s\circ\gamma' = \gamma_s\circ\alpha_s\circ \gamma = \gamma$.
\end{proof}

The following lemma is useful in the analysis of conditional and loop constructs in our quantum while language.

\begin{lemma}\label{lem:spmw}
	Suppose $\a$ as an abstract domain of quantum states is well-structured. Then for any $R\subseteq \dhv$,
	\begin{align*}
		\alpha(\sem{\pmstm}(R)) &= \alpha(\sem{\assert P[\bar{q}]; S_1}(R)) \vee \alpha(\sem{\assert P^\bot[\bar{q}]; S_0}(R))
		\\
		\alpha(\sem{\pwstm}(R)) &= \bigvee_{i\geq 0} \alpha\left( \sem{\left(\assert P[\bar{q}]; S\right)^i; \assert P^\bot[\bar{q}]}(R)\right)
	\end{align*}
\end{lemma}
\begin{proof}
	We only prove for the conditional case; the loop one is similar. Let $T_1 = \assert P[\bar{q}]; S_1$ and $T_0= \assert P^\bot[\bar{q}]; S_0$. Then
	\begin{align*}
		\alpha(\sem{\pmstm}(R))
		&= \alpha(\left\{\sem{T_1}(\rho) + \sem{T_0}(\rho): \rho\in R \right\})\\
		&= \textstyle\bigvee \left\{\alpha(\sem{T_1}(\rho) + \sem{T_0}(\rho)): \rho\in R \right\}\\
		&= \textstyle\bigvee \left\{\alpha(\sem{T_1}(\rho)) \vee \alpha(\sem{T_0}(\rho)): \rho\in R \right\}\\ 
		&= \textstyle\bigvee \left\{\alpha(\sem{T_1}(\rho)): \rho\in R \right\} \vee \bigvee \left\{ \alpha(\sem{T_0}(\rho)): \rho\in R \right\}\\  
		& = \alpha(\sem{T_1}(R)) \vee \alpha(\sem{T_0}(R))
	\end{align*}
	where the second and the last equalities follow from the fact that $(\alpha, \gamma)$ forms a Galois embedding, and the third one from Eq.~\eqref{eq:wels}.
\end{proof}

For well-structured abstract domain $\a$, if we are given a proper definition for the abstract operator $\sem{e}$, which is assumed to be monotonic, of each basic command $e \in \{\sskip,\ \bar{q} := |0\>,\ \bar{q}\apply U,\ \assert{P[\bar{q}]}\}$, then
the abstract operator $\aS: \a\ra \a$ for any composite quantum program $S$ can be defined inductively as follows: for any $a\in \a$,
\begin{enumerate}
	\item $\amap{S_0; S_1}(a) \define \amap{S_1}\circ \amap{S_0}(a)$;
	\item $\amap{\pmstm}(a) \define \amap{\assert P[\bar{q}]; S_1}(a) \vee \amap{\assert P^\bot[\bar{q}]; S_0}(a)$; 
	\item $\amap{\pwstm}(a) \define \bigvee_{i\geq 0} \amap{(\assert P[\bar{q}]; S_1)^i; \assert P^\bot[\bar{q}]}(a)$.
\end{enumerate}
It is easy to check that the induced abstract operator $\amap{S}$ is monotonic for any $S$ as well. The following theorem shows that such defined abstract operators are sound (resp. complete) if they are sound (resp. complete) for basic commands.

\begin{theorem}\label{thm:basicsuff}
	Let $\a$ be a well-structured abstract domain of quantum states. 
	\begin{enumerate}
		\item If $\asem{e}$ is sound for all basic commands $e$, then $\asem{S}$ is sound for any program $S$.
		\item If $\asem{e}$ is complete for all basic commands $e$, then $\asem{S}$ is complete for any program $S$.
	\end{enumerate}
\end{theorem}
\begin{proof}
	For clause (1), it suffices to prove by induction on the structure of $S$ that $\alpha\circ \sem{S} \leq_\a \aS \circ \alpha$ for any program $S$. Note that here we lift the order $\leq_\a$ between elements of $\a$ to functions from $\qconc$ to $\a$ in an entry-wise way. The basis case is directly from the assumption.
	\begin{enumerate}
		\item Let $S \equiv S_0; S_1$. By induction we have $\alpha\circ \sem{S_i} \leq_\a \amap{S_i} \circ \alpha$ for $i=0,1$. Then from the monotonicity of $\sem{S_0}$ and $\asem{S_1}$,
		$$\alpha\circ \sem{S} = \alpha\circ \sem{S_1}\circ \sem{S_0} \leq_\a \amap{S_1} \circ \alpha\circ \sem{S_0}  \leq_\a  \amap{S_1}\circ \amap{S_0} \circ \alpha = \aS \circ \alpha.$$
		
		\item Let $S\equiv \pmstm$. Let $T_1 = \assert P[\bar{q}]; S_1$ and $T_0= \assert P^\bot[\bar{q}]; S_0$. Then by induction, we have
		$\alpha\circ \sem{T_i} \leq_\a \amap{T_i} \circ \alpha$ for $i=0,1$. Then from Lemma~\ref{lem:spmw}, for any $R\subseteq \dhv$,
		\begin{align*}
			\alpha\circ \sem{S}(R) &= \alpha\circ \sem{T_1}(R) \vee  \alpha\circ \sem{T_0}(R)\\
			&\leq_\a \amap{T_1} \circ \alpha(R) \vee \amap{T_0} \circ \alpha(R) =  \aS\circ \alpha(R).	
		\end{align*}
		
		\item Let $S\equiv \pwstm$. Let $T = \assert P[\bar{q}]; S$ and $T_0 = \assert P^\bot[\bar{q}]$. Then by induction, we have
		$\alpha\circ \sem{T} \leq_\a \amap{T} \circ \alpha$ and $\alpha\circ \sem{T_0} \leq_\a \amap{T_0} \circ \alpha$. From Lemma~\ref{lem:spmw}, for any $R\subseteq \dhv$,
		\begin{align*}
			\alpha\circ \sem{S}(R) &= \bigvee_{i\geq 0} \alpha\circ \sem{T_0}\circ \sem{T}^i(R)\\
			&\leq_\a \bigvee_{i\geq 0}  \amap{T_0}\circ (\amap{T})^i\circ \alpha(R)=  \aS\circ \alpha(R).	  
		\end{align*}
	\end{enumerate}
	This completes the proof of clause (1). Clause (2) can be similarly proved.  
\end{proof}

The following example shows that the abstract domain $\shv$ allows every quantum program to have a complete abstraction.

\begin{example}\label{exm:soiscomplete} 
Consider the well-structured abstract domain $\shv$ of quantum states presented in Example~\ref{exm:allspaces}. Let us define for each basic command	the corresponding abstract operator as follows: for any $Q\in \shv$,
			\begin{align}
	\asem{\sskip}(Q) &\define Q\notag \\
	\asem{\bar{q}:=|0\>}(Q) &\define  \left\{|0\>_{\bar{q}}\otimes |\psi\> : |\psi\> \in \supp{\tr_{\bar{q}}(Q)}\right\} \label{eq:abasic}\\
	\asem{\bar{q}\apply U}(Q) &\define \left\{U_{\bar{q}}|\psi\> : |\psi\>\in Q \right\}\notag\\
	\asem{\assert{P[\bar{q}]}}(Q) &\define \spann\left\{P_{\bar{q}}|\psi\> : |\psi\>\in Q\right\}\notag. 
\end{align}
We would like to prove that these abstract operators are all complete. First, it is easy to check that al the sets on the right-hand side of Eq.~\eqref{eq:abasic} are valid subspaces of $\h_V$. Let us take $\assert{P[\bar{q}]}$ as an example. For any $R\in \qconc$,
\begin{align*}
	\alpha_s\circ \sem{\assert{P[\bar{q}]}}(R) & = \supp{\textstyle\bigcup \left\{P_{\bar{q}}\rho P_{\bar{q}} : \rho\in R\right\}}\\
	&	= \spann\left\{P_{\bar{q}}|\psi\> :  |\psi\>\in \supp{\rho}, \rho\in R\right\}.
\end{align*}
On the other hand, let $Q' \define  \textstyle\bigvee \left\{\supp{\rho} : \rho\in R\right\}$. Then
\begin{align*}
	\asem{\assert{P[\bar{q}]}}\circ \alpha_s(R) & = \asem{\assert{P[\bar{q}]}}(Q')\\
	&	= \spann\left\{P_{\bar{q}}|\phi\> :  |\phi\>\in  Q'\right\}.
\end{align*}
Note that any $ |\phi\>$ in $Q'$ can be written as a linear combination $|\phi\> =  \sum_i \beta_i |\psi_i\>$ where each $|\psi_i\>$ is taken from the support subspace of some state in $R$; that is, $|\psi_i\>\in \supp{\rho_i}$ while $\rho_i\in R$. Thus $$P_{\bar{q}}|\phi\>\in \spann\left\{P_{\bar{q}}|\psi\> :  |\psi\>\in \supp{\rho}, \rho\in R\right\},$$
and consequently, $\asem{\assert{P[\bar{q}]}}\circ \alpha_s(R)\subseteq \alpha_s\circ \sem{\assert{P[\bar{q}]}}(R)$. The other direction of inclusion is obvious.

From Theorem~\ref{thm:basicsuff}, such defined abstract operators can be extended to any composite quantum programs, and the extended ones are also complete.
In the following, we show a more direct way to define these complete abstract operators. Note that for any quantum program $S$, the denotational semantics $\sem{S}$ can be regarded as a completely positive and trace non-increasing super-operator over the set $\dhv$ of partial density operators. Thus by Kraus representation theorem~\cite{kraus1983states}, there exist a finite set of Kraus operators $E_k, k\in K$, such that for any $\rho\in \dhv$, $\sem{S}(\rho) = \sum_k E_k\rho E_k^\dag$. The abstract operator corresponding to $S$ can then be defined as follows:
\[
\asem{S}(Q) \define  \spann\{E_k |\psi\> : k\in K,|\psi\>\in Q\}. 
\]
Furthermore, this definition coincides with the abstract operators defined in Eq.~\eqref{eq:abasic}. Note that 
$$\supp{\sem{S}(\rho)} = \supp{\sum_k E_k\rho E_k^\dag} = \spann\left\{E_k|\psi\> : k\in K, |\psi\>\in \supp{\rho}\right\}.$$
Following similar lines of the proof for completeness of $\asem{\assert{P[\bar{q}]}}$ above, we can show that 
\[	\alpha_s\circ \sem{S}(R) =	\asem{S}\circ \alpha_s(R)  \]
for any program $S$ and set of states $R$. In other words, $\asem{S}$ is indeed the complete abstraction of $\sem{S}$.
\end{example}

To conclude this subsection, we show a useful property of complete abstraction of functions on quantum states.

\begin{lemma}\label{lem:fshapcon}
	Let $\a$ be a well-structured abstract domain for quantum states. 
	Let $a_i\in \a$ for each $i$, and $f^{\#}: \a\ra \a$ be the complete abstraction of operator $f: \dhv \ra \dhv$. Then 
	\[
	f^{\#}(\bigvee_{i} a_i) = \bigvee_{i} f^{\#}(a_i). 
	\]
\end{lemma}
\begin{proof}
	Let $R_i = \gamma(a_i)$. Then from the assumption that $(\alpha, \gamma)$ forms a Galois embedding between $\qconc$ and $\a$, we have $a_i = \alpha(R_i)$. Note that $\alpha(\textstyle\bigcup_i R_i) = \bigvee_i \alpha(R_i)$. Thus
	\begin{align*}
		f^{\#}(\textstyle\bigvee_{i} \alpha(R_i)) & =	f^{\#}(\alpha(\textstyle\bigcup_{i} R_i)) 
		= \alpha(f(\bigcup_{i} R_i))\\
		& = \alpha(\textstyle\bigcup_{i} f(R_i)) = \bigvee_{i} \alpha(f(R_i))\\
		& = \textstyle\bigvee_{i} f^{\#}(\alpha(R_i)) .  \qedhere
	\end{align*}
\end{proof}

\subsection{Well-structured quantum assertions}

Following the common practice of quantum Hoare logics in the literature, for the purpose of verification we only assume a semantic set of assertions $\a$ for quantum states and a \emph{satisfaction} relation $\models$ on $\dhv \times \a$. 
However, we do assume some structure of $\a$. Firstly, let a partial order $\leq_\a$ on $\a$ be defined as follows: $a\leq_\a a'$ iff for any $\rho\in \dhv$, $\rho\models a$ implies $\rho\models a'$. Furthermore, to describe conjunction and disjunction of assertions, we assume that $\a$ constitutes a complete lattice and let the meet and join be denoted by $\wedge$ and $\vee$, respectively. Finally, we make some assumptions on $\models$ to reflect the linear structure of $\dhv$.

\begin{definition}\label{def:wsassert}
		The complete lattice $\a$ as a set of quantum assertions for $\dhv$ is said to be \emph{well-structured}, whenever
	\begin{enumerate}
		\item if $\rho\models \qassert$ for all $\qassert\in A$ where $A\subseteq \a$, then $\rho\models \bigwedge A$; and 
		\item for any $a\in \a$, $\rho_i\in \dhv$, and $x_i >0$ with $\sum_{i}x_i\rho_i \in \dhv$,
		\begin{equation*}
			\sum_{i}x_i\rho_i\models a\quad \mbox{iff}\quad \forall i, \rho_i\models a.
		\end{equation*}
	\end{enumerate}
\end{definition}


\begin{example}\label{exm:assertallspaces}
	Recall the complete lattice
	\[
	(\shv, \subseteq, \vee, \cap, \{0\},  \h_V)
	\]
	 of all subspaces of $\h_V$	defined in Example~\ref{exm:allspaces}. We have shown that it is a well-structured abstract domain for quantum states in $\dhv$. Now we show that it can also serve as a well-structured set of quantum assertions by naturally defining the satisfaction relation $\models$ as follows: for any $\rho\in \dhv$ and $P\in \shv$,
	 \[
	 \rho\models P \quad \mbox{iff} \quad \supp{\rho}\subseteq P.
	 \]  
	 Firstly, for any subspaces $P$ and $Q$, $P\subseteq Q$ iff for any $\rho\in \dhv$, $\supp{\rho}\subseteq P$ implies $\supp{\rho}\subseteq Q$. Secondly,
	 if $\supp{\rho}\subseteq P$ for all $P\in A$ where $A$ is a set of subspaces in $\shv$, then obviously $\supp{\rho}\subseteq \bigwedge A$ as well. Finally, for any $P\in \shv$, $\rho_i\in \dhv$, and $x_i >0$ with $\sum_{i}x_i\rho_i \in \dhv$,
	 \begin{equation*}
	 	\supp{\sum_{i}x_i\rho_i} \subseteq P\quad \mbox{iff}\quad \forall i, \supp{\rho_i}\subseteq P.
	 \end{equation*}  
 Thus $\shv$ as a set of quantum assertions is indeed well-structured.
\end{example}
\subsection{Well-structured abstract domains v.s. well-structured assertions}

Now we examine the relationship between well-structured abstract  domains and assertion sets for quantum states. Firstly, we show how to transform a well-structured abstract quantum domain into a well-structured assertion set for quantum states.

\begin{lemma}\label{lem:wsadtoassert}
	Let $(\a, \leq_\a, \vee, \wedge, \bot, \top)$ be a well-structured abstract domain for quantum states, with $\alpha: \qconc\ra \a$ and $\gamma: \a\ra \qconc$ being the abstraction and concretisation functions respectively. Then the satisfaction relation $\models$ on $\dhv\times \a$, defined as for any $a\in \a$ and $\rho\in \dhv$,
	\[
	\rho\models a\quad \mbox{ iff }\quad \rho\in \gamma(a),
	\]
	turns $\a$ into a well-structured set of assertions for quantum states.
\end{lemma}
\begin{proof}
	First, from the fact that $(\alpha, \gamma)$ forms a Galois embedding between $\qconc$ and $\a$, for any $a,a'\in \a$,
	\[
	a\leq_{\a} a' \quad \Leftrightarrow\quad \gamma(a)\subseteq \gamma(a')  \quad \Leftrightarrow\quad \forall \rho, \rho\models a \mbox{ implies }\rho\models a'.
	\]
	We now prove that the two conditions in Definition~\ref{def:wsassert} are satisfied. To show (1), we note that for any $A\subseteq \a$ and $\rho\in \dhv$,
	\begin{align*}
		\forall a\in A, \rho \models a \quad &  \Rightarrow \quad	\forall a\in A, \rho \in \gamma(a)\\
		& \Rightarrow \quad \forall a\in A, \alpha(\rho) \leq_\a a \hspace{4em} \text{(Galois connection)}\\
		& \Rightarrow \quad \alpha(\rho) \leq_\a \textstyle\bigwedge A\\
		& \Rightarrow \quad \rho \in \gamma\left(\textstyle\bigwedge A\right) \hspace{6.6em} \text{(Galois connection)}\\
		& \Rightarrow \quad \rho\models \textstyle\bigwedge A.
	\end{align*}
	
	Furthermore, (2) is directly from Lemma~\ref{lem:spbased}(2).
\end{proof}

Conversely, a well-structured set of quantum assertions can be easily transformed into a well-structured abstract domain for quantum states as well.

\begin{lemma}\label{lem:wsasserttoad}
	Let $(\a, \leq_\a, \vee, \wedge, \bot, \top)$ be a well-structured set of quantum assertions, with $\models$ on $\dhv\times \a$ being the satisfaction relation. Then the pair of functions $\alpha: \qconc\ra \a$ and $\gamma: \a\ra \qconc$, defined as for any $a\in \a$ and $R\subseteq \dhv$,
	\begin{align*}
		\gamma(a) &\define \left\{\rho\in \dhv: \rho\models a\right\}\\
		\alpha(R) &\define \bigwedge \{b\in \a: R\subseteq \gamma(b)\}.
	\end{align*}
	turn $\a$ into a well-structured abstract domain for quantum states.
\end{lemma}
\begin{proof}
	We have to prove that the two conditions in Definition~\ref{def:wsdomain} are satisfied. First, from the assumption that $\models$ is consistent with $\leq_\a$, it is easy to prove that both $\gamma$ and $\alpha$ are monotonic functions. Second, to show $(\alpha, \gamma)$ forms a Galois connection between $\qconc$ and $\a$, it suffices to prove that for any $a\in \a$ and $R\subseteq \dhv$, $R\subseteq \gamma(a)$ iff $\alpha(R)\leq_\a a$.
		Let $A_R = \{b\in \a: R\subseteq \gamma(b)\}$. Then	
		\begin{align*}
			R\subseteq \gamma(a) \quad & \Rightarrow \quad a\in A_R\\
			& \Rightarrow \quad \alpha(R) = \bigwedge A_R \leq_\a a.
		\end{align*}
		Conversely, suppose $\alpha(R)\leq_\a a$. For any $\rho\in R$ and $b\in A_R$, we have $\rho\in \gamma(b)$, so $\rho\models b$. Thus from the fact that $\a$ as an assertion set is well-structured, $\rho\models \bigwedge A_R=\alpha(R)$ as well. By the assumption $\alpha(R)\leq_\a a$, we have $\rho\models a$, and so $\rho\in \gamma(a)$ as desired. Finally, for any $a$, we can show that $\alpha(\gamma(a)) =a$ from the fact that the satisfaction relation $\models$ is consistent with the partial order $\leq_\a$.
		
	 From the fact that $(\alpha, \gamma)$ forms a Galois embedding between $\qconc$ and $\a$, the remaining part of the proof is then directly from Lemma~\ref{lem:spbased}.
\end{proof}

\begin{example}
	We have already shown in Examples~\ref{exm:allspaces} and \ref{exm:assertallspaces} that the complete lattice
	\[
	(\shv, \subseteq, \vee, \cap, \{0\},  \h_V)
	\]
	can be both a well-structured abstract domain and a well-structured set of assertions  for quantum states. Actually, it can be easily seen that the Galois embedding $(\alpha_s, \gamma_s)$ defined in Example~\ref{exm:allspaces} and the satisfaction relation $\models$ defined in Example~\ref{exm:assertallspaces} satisfy the transformations stated in Lemmas~\ref{lem:wsadtoassert} and \ref{lem:wsasserttoad}.
\end{example}

%

\section{Quantum Hoare logic v.s. Abstract Interpretation}
\label{sec:hlvsai}

This section is devoted to the relationship between Hoare/incorrectness logic and abstract interpretation for quantum programs written in the while-language presented in Sec.~\ref{sec:qlang}.

\subsection{Hoare logic induced by abstract interpretation}

Let $\a$ be a well-structured abstract domain of program states, and an abstract monotonic operator $\asem{e}$ be defined for each basic command $e \in \{\sskip,\ \bar{q} := |0\>,\ \bar{q}\apply U,\ \assert{P[\bar{q}]}\}$. Then a Hoare-type proof system is naturally induced as follows:
\begin{enumerate}
	\item Take $\a$ to be the set of assertions. Furthermore, for any $\rho\in \d(\h_V)$ and $a\in \a$, $\rho\models a$ iff $\rho\in \gamma(a)$. From Lemma~\ref{lem:wsadtoassert}, $\a$ as a set of assertions is also well-structured.
	\item A correctness formula $\ass{\qassert}{S}{\qassertp}$ is valid, denoted $\models\ass{\qassert}{S}{\qassertp}$, if  $\sem{S}(\gamma(\qassert)) \subseteq \gamma(\qassertp)$; that is, $\gamma(\qassertp)$ is an over-approximation of $\sem{S}(\gamma(\qassert))$. From the assumption that $(\alpha, \gamma)$ forms a Galois embedding, this is equivalent to $\bsem{S}(a) \leq_\a b$ where $\bsem{S} = \alpha\circ\sem{S}\circ \gamma$ is the best abstraction of $\sem{S}$ in $\a$.
	\item The proof system (for partial correctness) is presented as in Table~\ref{tbl:psystem}. A correctness formula $\ass{\qassert}{S}{\qassertp}$ is said to be derivable, denoted $\vdash \ass{\qassert}{S}{\qassertp}$, if it has a proof sequence in the logic system. 
\end{enumerate}

Recall that such a proof system is said to be \emph{sound} if 
$\vdash\ass{\qassert}{S}{\qassertp}$ implies $\models\ass{\qassert}{S}{\qassertp}$
for any correctness formula $\ass{\qassert}{S}{\qassertp}$; while it is \emph{relatively complete} if the other direction of implication holds. 
The following theorem gives a close relationship between an abstract interpretation and the Hoare-type logic system induced by it. 

{\renewcommand{\arraystretch}{2.5}
	\begin{table}[t]
		\begin{lrbox}{\tablebox}
			\centering
			\begin{tabular}{l}
				\begin{tabular}{lc}
					(Exp)	&\qquad \qquad $\ass{\qassert}{e}{ \asem{e}(\qassert)}$\qquad where $e \in \{\sskip,\ \bar{q} := |0\>,\ \bar{q}\apply U,\ \assert{P[\bar{q}]}\}$
				\end{tabular}\\
				\begin{tabular}{lclc}
					(Seq)	&
					$\displaystyle\frac{\ass{\qassert}{S_0}{\qassert'},\ \ass{\qassert'}{S_1}{\qassertp}}{\ass{\qassert}{S_0; S_1}{\qassertp}}$ &
					(Meas)	&
					$\displaystyle\frac{\ass{\qassert}{\assert P[\bar{q}]; S_1}{\qassertp_1},\  \ass{\qassert}{\assert P^\bot[\bar{q}]; S_0}{\qassertp_0}}{\ass{\qassert}{\pmstm}{\qassertp_0\vee \qassertp_1}}$ \\
					(Imp)	&
					$\displaystyle\frac{\qassert\leq_\a \qassert',\ \ass{\qassert'}{S}{\qassertp'},\ \qassertp'\leq_\a \qassertp}{\ass{\qassert}{S}{\qassertp}}$&(While)	& $\displaystyle\frac{\ass{\qassert}{\assert P[\bar{q}]; S}{\qassert},\ \ass{\qassert}{\assert P^\bot[\bar{q}]}{\qassertp}}
					{\ass{\qassert}{\pwstm}{\qassertp}}$
				\end{tabular}\\
			\end{tabular}
		\end{lrbox}
		\resizebox{\textwidth}{!}{\usebox{\tablebox}}\\
		\vspace{4mm}
		\caption{Proof system for partial correctness induced by abstraction domain $(\a, \leq_\a)$. 
		}
		\label{tbl:psystem}
	\end{table}
}

\begin{theorem}\label{thm:aitohl}
	Let $\a$ be a well-structured abstract domain of quantum states.
	\begin{enumerate}
		\item If the abstract operator $\asem{e}$ is sound for each basic command $e$, then the induced proof system presented in Table~\ref{tbl:psystem} is sound. 
		\item If the abstract operator $\asem{e}$ is complete for each basic command $e$, then the induced proof system is both sound and relatively complete.
	\end{enumerate}
\end{theorem}
\begin{proof}
	For the first part, we have to prove that whenever $\vdash \ass{\qassert}{S}{\qassertp}$, then $\sem{S}(\gamma(\qassert)) \subseteq \gamma(\qassertp)$ or equivalently, $\bsem{S}(a) \leq_\a b$. This can be done by induction on the proof length of $\vdash \ass{\qassert}{S}{\qassertp}$. For the last step of the proof, we have the following cases:
	\begin{enumerate}
		\item Rule (Exp) is used. In this case, $S\equiv e$ for some basic command $e$, and $\qassertp \equiv \asem{e}(\qassert)$. The result then follows from the assumption that $\asem{e}$ is sound; that is, $\sem{e}(\gamma(a))\subseteq \gamma(\asem{e}(a))$.
		
		\item Rule (Seq) is used. By induction, $\sem{ S_0}(\gamma(\qassert))\subseteq \gamma(\qassert')$ and 
		$\sem{S_1}(\gamma(\qassert'))\subseteq \gamma(b)$. Thus
		\[
		\sem{S_0;S_1}(\gamma(\qassert)) =\sem{S_1}(\sem{S_0}
	(\gamma(\qassert)))  \subseteq \gamma(b).
		\]
		
		\item Rule (Imp) is used. Then the result follows from the monotonicity of $\bsem{S}$ for all program $S$.
				
		\item Rule (Meas) is used. Let $T_1 = \assert P[\bar{q}]; S_1$ and $T_0= \assert P^\bot[\bar{q}]; S_0$. By induction, we have $\bsem{ T_i}(\qassert)\leq_\a b_i$ for $i=0,1$. Thus from Lemma~\ref{lem:spmw}, 
		\[
		\bsem{\pmstm}(a) = \bsem{ T_0}(\qassert)\vee \bsem{ T_1}(\qassert)\leq_\a b_0\vee b_1.
		\]
		
			\item Rule (While) is used. From induction hypothesis, we have 
			\[
			\bsem{\assert P[\bar{q}]; S}(a)\leq_\a a, \quad \bsem{\assert P^\bot[\bar{q}]}(a)\leq_\a \qassertp.\]
			Let $T_i = (\assert P[\bar{q}]; S)^i; \assert P^\bot[\bar{q}]$ where $i=0,1,\ldots$. By induction on $i$ we can show
			$
			\bsem{T_i}(a)\leq_\a b$ for all $i\geq 0$.
			 Thus from Lemma~\ref{lem:spmw},  
			 		\[
			 \bsem{\pwstm}(a) = \bigvee_{i\geq 0} \bsem{T_i}(a) \leq_\a b.
			 \]
	\end{enumerate}

For the second part, suppose $\asem{e}$ is complete for any basic command $e$. Then from Theorem~\ref{thm:basicsuff}, the induced abstract operator $\asem{S}$ is complete for any quantum program $S$. Thus it must be the best abstraction; that is $\asem{S} = \bsem{S}$. Now we have to show that whenever 
$\bsem{S}(\qassert) \leq_\a \qassertp$, then $\vdash \ass{\qassert}{S}{\qassertp}$. 
From Rule (Imp), it suffices to show 
\begin{equation*}
\vdash \ass{\qassert}{S}{\bsem{S}(\qassert)}
\end{equation*}
 by induction on the structure of $S$. The basis case where $S\equiv e$ for some $e \in \{\sskip,\ \bar{q} := |0\>,\ \bar{q}\apply U,\ \assert{P[\bar{q}]}\}$ is directly from Rule (Exp). For other cases,
	\begin{enumerate}
	\item $S\equiv S_0;S_1$. This follows from the fact that complete abstractions are closed under operator composition; that is, $\bsem{S_0;S_1} = \bsem{S_1}\circ\bsem{S_0}$.
		
	\item $S\equiv \pmstm$. Let $T_1 = \assert P[\bar{q}]; S_1$ and $T_0= \assert P^\bot[\bar{q}]; S_0$. Then by induction, we have
	\[
	\vdash \ass{\qassert}{T_1}{\bsem{T_1}(\qassert)}, \qquad 	\vdash \ass{\qassert}{T_0}{\bsem{T_0}(\qassert)}.
	\] 
	Furthermore, from Lemma~\ref{lem:spmw},
	\begin{align*}
		\bsem{S}(a)  & = \alpha\circ \sem{T_1}\circ \gamma(a) \vee \alpha\circ \sem{T_0} \circ \gamma(a) = \bsem{T_1}(a) \vee \bsem{T_0}(a).
	\end{align*}
Thus the result follows from Rule (Meas).
	
	\item $S\equiv \pwstm$. Let $T = \assert P[\bar{q}]; S$, $T_0 = \assert P^\bot[\bar{q}]$, and $a^* = \bigvee_{i\geq 0} (\bsem{T})^i(a)$. Then by induction, 
	\[
	\vdash \ass{a^*}{T}{\bsem{T}(a^*)}, \qquad \vdash \ass{a^*}{T_0}{\bsem{T_0}(a^*)}.
	\]	
	From Lemma~\ref{lem:fshapcon}, we have $\bsem{T}(a^*) = \bigvee_{i\geq 0} ( \bsem{T})^{i+1}(a)\leq_\a a^*$ and 
	\[
		\bsem{T_0}(a^*) =  \bigvee_{i\geq 0} \bsem{T_0}\circ (\bsem{T})^i(a) = \bsem{S}(a)
	\] 
	where the second equality is from Lemma~\ref{lem:spmw}.
	The result then follows from Rules (Imp), (While), and the fact that $a\leq_\a a^*$. \qedhere
\end{enumerate}
\end{proof}

\begin{example}[Hoare logic induced by subspace abstract interpretation]\label{exm:applied}		
	
	Recall from Examples~\ref{exm:allspaces} and~\ref{exm:soiscomplete} that the subspace abstract domain $\shv$ for quantum states is well-structured, and the abstract operators for basic commands defined in Eq.~\eqref{eq:abasic} are complete. Thus by Theorems~\ref{thm:aitohl} the induced Hoare-type proof system presented in Table~\ref{tbl:psystem} is both sound and relatively complete when assertions are taken from $\shv$. 
	
	Note that the applied quantum Hoare logic proposed in~\cite{zhou2019applied} also uses elements of $\shv$ as assertions. 
	A correctness formulas $\ass{P}{S}{Q}$ is correct (with respect to partial correctness), denoted $\models_{\text{ap}} \ass{P}{S}{Q}$, if for any $\rho\in \dhv$, whenever $\supp{\rho} \subseteq P$, $\supp{\sem{S}(\rho)} \subseteq Q$. Using the concretisation function $\gamma_s$ from $\shv$ to $\qconc$ defined in Example~\ref{exm:allspaces}, this is equivalent to $\sem{S}(\gamma_s(P)) \subseteq \gamma_s(Q)$, coinciding with our correctness definition. 
	
	We now compare our proof system with the applied quantum Hoare logic. First,	let us examine the inference rules for initialisation $\bar{q}:=|0\>$:
	\[
	\text{(Init)}\quad \displaystyle \ass{P}{\bar{q}:=|0\>}{|0\>_{\bar{q}}\<0|\otimes \supp{\tr_{\bar{q}}(P)}}
	\qquad \quad
	\text{(Init-ap)}\quad \displaystyle \ass{I_{\bar{q}}\otimes f(Q)}{\bar{q}:=|0\>}{Q}.
	\]
	The left one is from our system, while the right one is taken from~\cite{zhou2019applied} where
	  \[ f(Q) \define \textstyle \bigvee \{T\in \shv : |0\>_{\bar{q}}\<0|\otimes T\subseteq Q\}.\]
	  We now prove that with the help of Rule (Imp), these two rules are indeed equivalent:
	  \begin{enumerate}
	  	\item (Init) $\Rightarrow$ (Init-ap). Suppose $Q$ is given. Let $P \define I_{\bar{q}}\otimes f(Q)$. Then $\supp{\tr_{\bar{q}}(P)}= f(Q)$. Thus from (Init) we have $\ass{P}{\bar{q}:=|0\>}{|0\>_{\bar{q}}\<0|\otimes f(Q)}$, and so (Init-ap) follows from (Imp) by noting that $|0\>_{\bar{q}}\<0|\otimes f(Q)\subseteq Q$.
	  	\item (Init-ap) $\Rightarrow$ (Init). Suppose $P$ is given. Let $Q \define |0\>_{\bar{q}}\<0|\otimes \supp{\tr_{\bar{q}}(P)}$. Then $f(Q)=\supp{\tr_{\bar{q}}(P)}$. Thus from (Init-ap) we have $\ass{I_{\bar{q}}\otimes \supp{\tr_{\bar{q}}(P)}}{\bar{q}:=|0\>}{Q}$, and so (Init) follows from (Imp) by noting that $P\subseteq I_{\bar{q}}\otimes \supp{\tr_{\bar{q}}(P)}$.
	  \end{enumerate}
  	 
  	 It is evident that the proof system of~\cite{zhou2019applied} works in a backward manner; that is, it tries to give \emph{the weakest precondition} of a postcondition. In contrast, our logic system in Table~\ref{tbl:psystem} works in a forward manner by trying to give \emph{the strongest postcondition} of a precondition. 
	  Due to this complementary nature, we believe that these two systems will be useful in different applications. 
\end{example}


Next, let us examine the abstract interpretation for quantum circuits proposed in~\cite{yu2021quantum} where tuples of subspaces of some subsystems, instead of subspaces of the whole system, are regarded as abstract elements for quantum states. 

\begin{example}[Local-subspace abstract domain]\label{exm:yuai}
	Let a signature $\sigma$ be a tuple of proper subsets of $V$; formally, $ \sigma \define (s_1, \cdots, s_m)$ where $m\geq 1$ and for each $i$, $s_i \subset V$. Then the abstract subspace domain with signature $\sigma$ is given by
	\[
	(\shv_\sigma, \sqsubseteq_\sigma, \sqcup_\sigma, \sqcap_\sigma, \bot_\sigma, \top_\sigma)
	\]
	where 
	\[
	\shv_\sigma \define \left\{(P_1, \cdots, P_m) : \forall i, P_i \mbox{ is a subspace of } \h_{s_i}\right\},
	\]
	the partial order $\sqsubseteq_\sigma$ and the lattice operators $\sqcup_\sigma$ and $\sqcap_\sigma$ are all defined in the entry-wise way, $\bot_\sigma \define (0_{s_1}, \cdots, 0_{s_m})$, and $\top_\sigma\define (I_{s_1}, \cdots, I_{s_m})$. Again, this domain is a complete lattice. When each $s_i$ contains exactly two quantum variables, this is similar to the octagon domain for classical programs.
	
	The abstraction and concretisation functions are defined naturally as follows.
	For any $\widetilde{P} \define (P_1, \cdots, P_m)\in \shv_{\sigma}$ and $R\in \qconc$,
	\begin{align*}
		\alpha_\sigma(R) & \define (Q_1, \cdots, Q_m),
\mbox{ where $ Q_i \define \bigvee\left\{\supp{\rho_i} : \rho \in R\right\}$ and }\\
&\hspace{7.5em}  \rho_i \define \tr_{V\backslash s_i}(\rho) \mbox{ is the reduced state of $\rho$ in the subsystem $s_i$;}\\
		\gamma_\sigma(\widetilde{P}) & \define \bigcap_{i=1}^m \gamma_s(P_i) = \bigcap_{i=1}^m \left\{\rho\in \dhv : \supp{\rho}\subseteq P_i\otimes I_{V\backslash s_i}\right\}.
	\end{align*}
	Intuitively, $\alpha_\sigma(R)$ is the tuple of subspaces in which each component subspace is spanned by all the quantum states in $R$ restricted on the corresponding subsystems, while $\gamma_\sigma(\widetilde{P})$ is collectively determined by all the local subspaces $P_i$ in $\widetilde{P}$ when $P_i$ is regarded as a subspace of the whole quantum system $V$. In other words, $\shv_\sigma$ is essentially the direct product of the individual abstract domains $\s(\h_{s_i})$.
	
	We now show that for each signature $\sigma$,  $\shv_\sigma$ is well-structured. First, for any $\widetilde{P}\in \shv_{\sigma}$ and $R\in \qconc$,
	\begin{align*}
		\alpha_\sigma(R) \le_\sigma \widetilde{P}\quad &\Leftrightarrow \quad 
		\forall i, \textstyle \bigvee\left\{\supp{\rho_i} : \rho \in R\right\} \subseteq P_i\\
		 &\Leftrightarrow\quad \forall i, \forall \rho\in R,  \supp{\tr_{V\backslash s_i}(\rho)}\subseteq P_i\\
		 &\Leftrightarrow\quad \forall \rho\in R, \forall i,   \supp{\rho}\subseteq P_i\otimes I_{V\backslash s_i}\\	
		 &\Leftrightarrow \quad  R\subseteq \gamma_\sigma(\widetilde{P})		 
	\end{align*}
	where the third equivalence follows from the fact that $\supp{\rho} \subseteq \supp{\tr_{V\backslash s_i}(\rho)}\otimes I_{V\backslash s_i}$ and 
	$\supp{\tr_{V\backslash s_i}(\supp{\rho})} = \supp{\tr_{V\backslash s_i}(\rho)}$
	 for all $\rho\in \dhv$. Furthermore, $\alpha_\sigma$ is surjective. Thus the pair $(\alpha_\sigma, \gamma_\sigma)$ forms a Galois embedding between $\qconc$ and $\shv_\sigma$. Second, for any $k=1, \cdots, m$, $\rho_i\in \dhv$, and $x_i >0$ with $\sum_{i}x_i\rho_i \in \dhv$, 
	 \begin{equation*}
	 	\left(\alpha_\sigma\left(\sum_{i}x_i\rho_i\right) \right)_k= \supp{\tr_{V\backslash s_k}\left(\sum_{i}x_i\rho_i\right)} =\bigvee_i \supp{\tr_{V\backslash s_k}\left(\rho_i\right)} = \left(\bigvee_i \alpha_\sigma(\rho_i)\right)_k.
	 \end{equation*}  	 
	 
	It is obvious that the best abstraction $\sem{S}^b_\sigma = \alpha_\sigma\circ\sem{S}\circ \gamma_\sigma$ is sound for any quantum program $S$. Thus from Theorem~\ref{thm:aitohl}, the proof system presented in Table~\ref{tbl:psystem}, when $\asem{e}$ in Rule (Exp) is replaced by $\sem{e}^b_\sigma$, is sound for quantum assertions taken from $\shv_{\sigma}$.
	 However, as the following counter-example shows, $\sem{S}^b_\sigma$ is in general not a complete abstraction for $\sem{S}$. For simplicity, we consider the case where $V = \{q_1, q_2, q_3\}$ and $\sigma = (\{q_1, q_2\}, \{q_2, q_3\})$. Let $|\Phi^{\pm}\> \define \frac{1}{\sqrt{2}}(|000\> \pm |111\>)$, $U$ be a unitary operator on $\dhv$ which maps $|\Phi^+\>$ to $|000\>$ and $|\Phi^-\>$ to $|111\>$,  and $S\define  
	 q_1,q_2,q_3\apply U$. Let $\Phi^+ = |\Phi^+\>\<\Psi^+|$. Note that $\Phi^+$ can be regarded as either the density operator corresponding to the pure state $ |\Phi^+\>$ or the projector onto the one dimensional subspace spanned by $ |\Phi^+\>$.
	 Then
	 \[
	 \alpha_\sigma\circ \sem{S}(\Phi^+)  = (P^{00}_{q_1,q_2}, P^{00}_{q_2,q_3}) \quad\mbox{ and } \quad \alpha_\sigma(\Phi^+) = (P^{=}_{q_1,q_2}, P^{=}_{q_2,q_3})
	 \]
	 where $P^{ii} = \spann\{|ii\>\}$, $i=0,1$, and $P^{=} = \spann\{|00\>, |11\>\}$. However, since $\Phi^- \in \gamma_\sigma(\alpha_\sigma(\Phi^+) )$, we have
	 $$|111\>\<111|\in  \sem{S}\circ \gamma_\sigma\circ \alpha_\sigma(\Phi^+).$$
	 Thus if $
	 \sem{S}^b_\sigma \circ \alpha_\sigma(\Phi^+) = (Q_1, Q_2)$ then it must hold that
	 $|11\> \in Q_1$ and $|11\> \in Q_2$. Consequently, $\alpha_\sigma\circ \sem{S}(\Phi^+)  \neq  \sem{S}^b_\sigma \circ \alpha_\sigma(\Phi^+) $ as desired.
	 
	 As $\sem{S}^b_\sigma$ is not a complete, Theorem~\ref{thm:aitohl} tells us nothing about the completeness of the induced Hoare system. Actually, we are going to show in Sec.~\ref{subsec:hltoai} that there does not exist a relatively complete Hoare system which takes  $\shv_\sigma$ as the set of assertions.
\end{example}

\subsection{Incorrectness logic induced by abstract interpretation}

In the previous section, we show how an abstract interpretation of quantum programs is closely related to a quantum Hoare logic induced by it. Interestingly, a quantum incorrectness logic system can also be induced by a quantum abstract interpretation.

Let $\a$ be a well-structured abstract domain of program states, and an abstract monotonic operator $\asem{e}$ be defined for each basic command $e \in \{\sskip,\ \bar{q} := |0\>,\ \bar{q}\apply U,\ \assert{P[\bar{q}]}\}$. Then an incorrectness proof system is naturally induced as follows:
\begin{enumerate}
	\item Take $\a$ to be the set of assertions, and the satisfaction relation is similarly defined as for the Hoare logic case in the previous section. Again, $\a$ as a set of assertions is also well-structured.
	\item A specification formula $\iass{\qassert}{S}{\qassertp}$ holds, denoted $\models_{\text{in}}\iass{\qassert}{S}{\qassertp}$, if $b\leq_\a \bsem{S}(a)$. That is, $b$ is an under-approximation of 	$\alpha(\sem{S}\circ\gamma(a))$, the abstraction for the set of reachable states of $S$ starting in $\gamma(a)$.
	\item The proof system is presented as in Table~\ref{tbl:ipsystem}. Denote by $\vdash_{\text{in}}\iass{\qassert}{S}{\qassertp}$ if the formula $\iass{\qassert}{S}{\qassertp}$ can be deduced from it.
\end{enumerate}

Compared with the Hoare logic system in Table~\ref{tbl:psystem}, the main difference lies in the consequence rule. In Rule (Imp), we strengthen preconditions and weaken postconditions, while in Rule (Imp-In) we weaken preconditions and strengthen postconditions. This is due to the fact that Hoare logic reasons about over-approximation while incorrectness logic reasons about under-approximation of program semantics.

%

{\renewcommand{\arraystretch}{2.5}
	\begin{table}[t]
		\begin{lrbox}{\tablebox}
			\centering
			\begin{tabular}{l}
				\begin{tabular}{lc}
					(Exp-In)	&\qquad \qquad $\iass{\qassert}{e}{ \asem{e}(\qassert)}$\qquad where $e \in \{\sskip,\ \bar{q} := |0\>,\ \bar{q}\apply U,\ \assert{P[\bar{q}]}\}$
				\end{tabular}\\
				\begin{tabular}{lclc}
					(Seq-In)	&
					$\displaystyle\frac{\iass{\qassert}{S_0}{\qassert'},\ \iass{\qassert'}{S_1}{\qassertp}}{\iass{\qassert}{S_0; S_1}{\qassertp}}$ & \quad
					(Meas-In)	&
					$\displaystyle\frac{\iass{\qassert}{\assert P[\bar{q}]; S_1}{\qassertp_1},\  \iass{\qassert}{\assert P^\bot[\bar{q}]; S_0}{\qassertp_0}}{\iass{\qassert}{\pmstm}{\qassertp_0\vee \qassertp_1}}$ \\
					(Imp-In)	&
					$\displaystyle\frac{\qassert'\leq_\a \qassert,\ \iass{\qassert'}{S}{\qassertp'},\ \qassertp\leq_\a \qassertp'}{\iass{\qassert}{S}{\qassertp}}$&\quad(While-In)	& $\displaystyle\frac{\forall i, \iass{\qassert_i}{\assert P[\bar{q}]; S}{\qassert_{i+1}},\ \iass{\qassert_i}{\assert P^\bot[\bar{q}]}{\qassertp_i}}
					{\iass{\qassert_0}{\pwstm}{\bigvee_{i\geq 0} \qassertp_i}}$
				\end{tabular}\\
			\end{tabular}
		\end{lrbox}
		\resizebox{\textwidth}{!}{\usebox{\tablebox}}\\
		\vspace{4mm}
		\caption{Proof system for incorrectness logic induced by abstract domain $(\a, \leq_\a)$. 
		}
		\label{tbl:ipsystem}
	\end{table}
}

As abstract interpretation usually provides an over-approximation for program analysis, a sound abstraction does not necessarily guarantee the soundness of the induced incorrectness logic. However, as the following theorem states, a complete abstract interpretation indeed guarantees that the induced incorrectness logic is both sound and relatively complete.

\begin{theorem}\label{thm:aitoihl}
	Let $\a$ be a well-structured abstract domain of quantum programs. If the abstract operator $\asem{e}$ is complete for each basic command $e$,  then the induced incorrectness logic system presented in Table~\ref{tbl:ipsystem} is both sound and relatively complete; that is, $$\vdash_{\rm{in}}\iass{\qassert}{S}{\qassertp}\qquad \mbox{iff} \qquad  \models_{\rm{in}}\iass{\qassert}{S}{\qassertp}$$ for any specification formula $\iass{\qassert}{S}{\qassertp}$.
\end{theorem}
\begin{proof}
	First, from Theorem~\ref{thm:basicsuff}, if $\asem{e}$ is complete for each basic command $e$, then $\asem{S}$ is complete for any program $S$. Thus $\asem{S} = \bsem{S}$.

	To prove the soundness, we have to show that whenever $\vdash_{\text{in}} \iass{\qassert}{S}{\qassertp}$, then $b\leq_\a \bsem{S}(a)$. This can be proved by induction on the proof length of $\vdash_{\text{in}} \iass{\qassert}{S}{\qassertp}$. For the last step of the proof, we have the following cases:
	\begin{enumerate}
		\item Rule (Exp-In) is used. In this case, $S\equiv e$ for some $e \in \{\sskip,\ \bar{q} := |0\>,\ \bar{q}\apply U,\ \assert{P[\bar{q}]}\}$, and $\qassertp \equiv \bsem{e}(\qassert)$. The result then trivially follows. 
		
				\item Rule (Seq-In) is used. This follows from the fact that complete abstractions are closed under operator composition; that is, $\bsem{S_0;S_1} = \bsem{S_1}\circ\bsem{S_0}$.
		
		\item Rule (Imp-In) is used. Then the result follows from the monotonicity of $\bsem{S}$ for all program $S$.
		
		\item Rule (Meas-In) is used. Let $T_1 = \assert P[\bar{q}]; S_1$ and $T_0= \assert P^\bot[\bar{q}]; S_0$. By induction, we have $b_i \leq_\a \bsem{ T_i}(\qassert)$ for $i=0,1$. Thus from Lemma~\ref{lem:spmw}, 
		\[
		b_0\vee b_1\leq_\a
		\bsem{ T_0}(\qassert)\vee \bsem{ T_1}(\qassert)=\bsem{\pmstm}(a).
		\]
		
		\item Rule (While-In) is used. Let $T = \assert P[\bar{q}]; S$ and $T_0 = \assert P^\bot[\bar{q}]$. From induction hypothesis, we have for any $i\geq 0$,
		$
		a_{i+1} \leq_\a \bsem{T}(a_i)$ and  $\qassertp_i\leq_\a \bsem{T_0}(a_i).$
		By induction on $i$ we can show
		$a_i\leq_\a
		(\bsem{T})^i(a_0)$ for all $i\geq 0$.
		Thus from Lemma~\ref{lem:spmw},  
		\[
		\bigvee_{i\geq 0} b_i \leq_\a \bigvee_{i\geq 0} \bsem{T_0}(a_i) \leq_\a \bigvee_{i\geq 0} \bsem{T_0}\circ	(\bsem{T})^i(a_0) = \bsem{\pwstm}(a_0).
		\]
	\end{enumerate}

For the completeness part, we have to show that whenever $\qassertp\leq_\a \bsem{S}(\qassert)$, then $\vdash_{\text{in}} \iass{\qassert}{S}{\qassertp}$. 
First, from Rule (Imp-In), it suffices to show 
\begin{equation*}
\vdash_{\text{in}} \iass{\qassert}{S}{\bsem{S}(\qassert)}
\end{equation*}
by induction on the structure of $S$. The proof is similar to the corresponding part of Theorem~\ref{thm:aitohl} for Hoare logic, except for the while loop which we prove as follows.

%
%
Let $S\equiv \pwstm$, $T = \assert P[\bar{q}]; S$, and $T_0 = \assert P^\bot[\bar{q}]$. By induction, for any $i\geq 0$,
\[
\vdash_{\text{in}} \iass{(\bsem{T})^i(a)}{T}{(\bsem{T})^{i+1}(a)}, \quad \vdash_{\text{in}} \iass{(\bsem{T})^i(a)}{T_0}{\bsem{T_0}\circ (\bsem{T})^i(a)}.
\]
		Thus from (While-In),
\[
\vdash_{\text{in}} \iass{a}{S}{\bigvee_{i\geq 0} \bsem{T_0}\circ (\bsem{T})^i(a)},
\]		
and the result follows from	Lemma~\ref{lem:spmw} and  the fact that complete abstractions are closed under operator composition. \qedhere
\end{proof}

\begin{example}[Quantum incorrectness logic induced by subspace abstract interpretation]	\label{exm:inclogic}
		
	Again, as the subspace abstract interpretation defined Example~\ref{exm:soiscomplete} is complete, from  Theorems~\ref{thm:aitoihl} we know that the induced incorrectness logic system presented in Table~\ref{tbl:ipsystem} is both sound and complete when assertions are taken from $\shv$. 
				
	The incorrectness logic proposed in~\cite{yan2022incorrectness} also uses elements of $\shv$ as assertions. A specification formula
	$\iass{P}{S}{Q}$ is valid if for any $\rho\in \dhv$, whenever $P \subseteq \supp{\rho}$, $Q\subseteq\supp{\sem{S}(\rho)}$.
	It is easy to see that this is equivalent to $Q\leq_\a \bsem{S}(P)$ where the best abstraction $\bsem{S}$ is defined using the abstraction and concretisation functions presented in Example~\ref{exm:allspaces}, coinciding with our correctness definition. Again, the difference between our proof system and the one proposed in~\cite{yan2022incorrectness} is that the former works in a forward manner while the latter in a backward manner.
\end{example}

\subsection{Abstract interpretation induced by Hoare logic}
\label{subsec:hltoai}
Now we consider the reverse problems of deriving a quantum abstract interpretation from a quantum Hoare/incorrectness logic. Let $\a$ be a well-structured set of quantum assertions, and $\ps$ be a Hoare-type proof system (of partial correctness) for quantum programs, where the assertions are taken from $\a$.
Without loss of generality, we assume that $\ps$ consists of an axiom (schema) for each basic command and a proof rule for each program construct such as sequential composition, conditional branching, and while loop. Furthermore, it provides a consequence rule of the form
\[
	\displaystyle\frac{\qassert\leq_\a \qassert',\ \ass{\qassert'}{S}{\qassertp'},\ \qassertp'\leq_\a \qassertp}{\ass{\qassert}{S}{\qassertp}}
\]
for precondition strengthening and postcondition weakening. 
A correctness formula $\ass{\qassert}{S}{\qassert'}$ in $\ps$ is semantically valid, written $\models_{\ps} \ass{a}{S}{a'}$,  if for any  $\rho\in \dhv$, $\rho\models \qassert$ implies $\sem{S}(\rho)\models \qassert'$. It is derivable, written $\vdash_{\ps} \ass{\qassert}{S}{\qassert'}$, if it has a proof sequence in $\ps$. Then an abstract interpretation for quantum programs is naturally induced by $\ps$ as follows:
%
\begin{itemize}
	\item Take $\a$ to be the abstract domain of quantum states, and define the pair of abstraction-concretisation functions $(\alpha, \gamma)$ as in Lemma~\ref{lem:wsasserttoad}. Then $\a$ as an abstract domain is also well-structured.
\item For any quantum program $S$, let $\asem{S}: \a\ra \a$ with
\begin{equation}\label{eq:defasem}
	\asem{S}(a) = \bigwedge \{a'\in \a :\ \vdash_{\ps} \ass{a}{S}{a'}\}
\end{equation}
for any $a\in \a$ be the abstract operator of $\sem{S}$.
\end{itemize}


Thanks to the complete lattice structure of the assertion set $\a$, the strongest postconditions always exist in $\ps$. For any quantum program $S$ and $a\in \a$, let
\begin{equation}\label{eq:defspc}
	spc(S,a) \define \bigwedge \{a'\in \a :\ \models_{\ps} \ass{a}{S}{a'}\}.
\end{equation}
We now prove that $spc(S,a)$ is the strongest postcondition of $a$ with respect to $S$. Furthermore, it coincides with the best abstraction of $S$.
\begin{lemma}\label{lem:spost}
For any quantum program $S$ and $a\in \a$,
	\begin{enumerate}
		\item $\models_{\ps} \ass{a}{S}{spc(S,a)}$;
		\item for any $a'\in \a$, if $\models_{\ps} \ass{a}{S}{a'}$ then $spc(S,a)\leq_\a a'$;
		\item $spc(S,a) = \bsem{S}(a)$ where $\bsem{S}=\alpha\circ\sem{S}\circ\gamma$ is the best abstraction of $\sem{S}$. 
	\end{enumerate}
\end{lemma}
\begin{proof}
	Firstly, clause (2) is easy from Eq.~\eqref{eq:defspc}. To prove (1), take any $\rho$ with $\rho\models a$ and $a'\in \a$. If $\models_{\ps} \ass{a}{S}{a'}$, then $\sem{S}(\rho)\models a'$. Thus $\sem{S}(\rho)\models spc(S,a)$ from the first condition of Definition~\ref{def:wsassert}. 

	To prove  (3), we first observe that $\models_{\ps} \ass{a}{S}{ \bsem{S}(a) }$. Thus $spc(S,a) \leq_\a \bsem{S}(a)$ from (2). For the converse part,
	take any $\qassert' \in \a$ with $\models_{\ps} \ass{a}{S}{a'}$. Then $\sem{S}\circ \gamma(\qassert) \subseteq \gamma(\qassert')$, which implies that $\bsem{S}(a)  \leq_\a \qassert'$. Thus $\bsem{S}(a)  \leq_\a spc(S,a)$ from the arbitrariness of $a'$.
\end{proof}

The following theorem gives a close relationship between a Hoare-type logic system and the abstract interpretation induced by it.

\begin{theorem}\label{thm:hltoai}
	Let $\a$ be a well-structured set of quantum assertions, and $\ps$ a Hoare-type logic system for quantum programs with assertions taken from $\a$.
	\begin{enumerate}
		\item If $\ps$ is sound, then the induced abstract interpretation is sound.
		\item If $\ps$ is sound and relatively complete, then the induced abstract interpretation is complete.
	\end{enumerate}
\end{theorem}
\begin{proof}
For any $a\in \a$ and quantum program $S$, let 
$$F_{\vdash}(S,a) \define \{a'\in \a :\ \vdash_{\ps} \ass{a}{S}{a'}\}\quad \mbox{and}\quad 
F_{\models}(S,a) \define \{a'\in \a :\ \models_{\ps} \ass{a}{S}{a'}\}.$$ 
Now to prove (1), note that if $\ps$ is sound then $F_{\vdash}(S,a) \subseteq F_{\models}(S,a)$. Thus  $\bsem{S}(a)\leq_\a \asem{S}(a)$ by Lemma~\ref{lem:spost}(3), which implies that the abstraction $\asem{S}$ is sound.

To prove (2), note that if $\ps$ is relatively complete then $F_{\models}(S,a) \subseteq F_{\vdash}(S,a)$. Thus  $\asem{S}(a)\leq_\a\bsem{S}(a)$, which, together with (1), implies that $\asem{S}$ is the best abstraction of $S$. The rest of the proof consists of four steps.

\begin{itemize}
	\item[(i)] As the first step, we prove that for any $S_2$ and $S_1$, 
	\[
	\asem{S_1; S_2} = \asem{S_2}\circ \asem{S_1}.
	\]
	For any $a\in \a$, we know from the completeness of $\ps$ that
	$\vdash_{\ps} \ass{a}{S_1;S_2}{\asem{S_1; S_2}(a)}$. Furthermore,
	as there is only one proof rule for sequential composition, there must exist
	$a_1, b\in \a$ such that
	\[
	a\leq_\a a_1,\quad \vdash_{\ps}\ass{a_1}{S_1}{b},\quad \vdash_{\ps} \ass{b}{S_2}{a'},\quad a'\leq_\a \asem{S_1; S_2}(a).
	\]
	Then 
	\begin{align*}
		\asem{S_2}\circ \asem{S_1}(a) & \leq_\a \asem{S_2}\circ \asem{S_1}(a_1)\\
		&\leq_\a \asem{S_2}(b)\\
		&\leq_\a a' \leq_\a \asem{S_1; S_2}(a).
	\end{align*}
	The other direction that $\asem{S_1; S_2}(a)\leq_\a \asem{S_2}\circ \asem{S_1}(a)$ is easy from the fact $\gamma\circ \alpha\geq_{\qconc} \mathrm{id}_{\qconc}$.
	
	\item[(ii)] The next step is to show that for any set $R$ of quantum states, we can always find a single state (not necessarily in $R$) which shares the same abstraction with $R$. Specifically, we have
	\begin{clm}\label{clm:settostate}
		For any $R\subseteq \dhv$,  there exists a single state $\rho_R\in \dhv$ with $\tr(\rho_R) = 1$ such that 
		\[\alpha(\sem{S}(\rho_R)) = \alpha(\sem{S}(R))\] for all quantum program $S$. In particular, $\alpha(\rho_R) = \alpha(R)$.
	\end{clm}
	\begin{proof}[Proof of Claim \ref{clm:settostate}]
		As $\h_V$ is finite dimensional, we can always find a set of states $\rho_1, \ldots, \rho_n$ in $R$, $\tr(\rho_i)=1$, such that 
		\[
			\alpha_s(R) = \bigvee \left\{\supp{\rho} : \rho \in R\right\} = \bigvee_{i=1}^n \supp{\rho_i}.
		\]
		Let $\rho_R = \sum_{i=1}^n \rho_i/n$. Then for any $S$, 
		\[
		\alpha_s(\sem{S}(\rho_R)) = \alpha_s\left(\sum_{i=1}^n\frac{1}{n} \sem{S}(\rho_i)\right) = \bigvee_{i=1}^n \alpha_s(\sem{S}(\rho_i)) = \alpha_s(\sem{S}(R)).
		\] 
		Then $\alpha(\sem{S}(\rho_R)) = \alpha(\sem{S}(R))$ from Lemma~\ref{lem:tmp}.
	\end{proof}
\item[(iii)]  The third step is to show that our while-language is powerful enough to prepare an arbitrarily given state. Specifically, we have
\begin{clm}\label{clm:singlestate}
	For any $\rho\in \dhv$ with $\tr(\rho) = 1$, there exists a quantum program $S^\rho$ which turns any quantum state into state $\rho$; that is, for all $\sigma\in \dhv$,
	\[
	\sem{S^\rho}(\sigma) = \tr(\sigma)\cdot \rho.
	\] 
	Consequently, for all $R\subseteq \dhv$ and program $S$, \[
	\alpha(\sem{S^\rho;S}(R)) =  \alpha(\sem{S}(\rho))
	\]
	whenever $R\neq \{0\}$. 
\end{clm}
\begin{proof}[Proof of Claim \ref{clm:singlestate}]
	Let $\rho = \sum_{i=0}^{d-1} \lambda_i |\psi_i\>\<\psi_i|$, $d=2^{|V|}$, be the spectral decomposition of $\rho$ where $\{|\psi_i\> : i=0,\ldots, d-1\}$ constitute some orthonormal basis of $\h_V$. 
	Let $|\psi\> =  \sum_{i=0}^{d-1} \sqrt{\lambda_i} |\psi_i\>$. From the assumption that $\tr(\rho) = 1$ we know that $|\psi\>$ is a valid quantum (pure) state. Thus we can find a unitary operator $U$ on $\h_V$ such that $U|0\> = |\psi\>$. Let 
	\begin{align*}
		S^{\rho} \define &\\
		& \bar{q} := |0\>; \ \bar{q} \apply U;\\	
		&\iif\ P_0[\bar{q}]\ \then\\
		&\qquad  \sskip;\\
		&\eelse\ \iif\ P_1[\bar{q}]\ \then\\
		&\qquad\qquad  \sskip;\\	
		&\qquad\eelse\\
		&\qquad\qquad \ldots \\
		& \qquad\pend\\
		& \pend
	\end{align*} 
	where $\bar{q}=V$, and for any $i=0, \ldots, d-1$, $P_i = |\psi_i\>\<\psi_i|$. Note that for any $\sigma\in \dhv$, $$\sem{\bar{q} := |0\>}(\sigma) = \tr(\sigma)\cdot |0\>\<0|.$$
	Then it is easy to show that $\sem{S^\rho}(\sigma) = \tr(\sigma)\cdot \rho$. The last part of the claim follows from the linearity of $\sem{S}$.
\end{proof}

\item[(iv)] Now for any quantum program $S$, $R\subseteq \dhv$, and $a\in \a$ with $\gamma(a) \neq \{0\}$, 
\begin{align*}
	\alpha\circ \sem{S}( R)	& =\alpha (\sem{S}( \rho_R)) \hspace{11em} \mbox{Claim~\ref{clm:settostate}}\\
	&=\alpha\circ \sem{S^{\rho_R};S}\circ \gamma(a) \hspace{8em} \mbox{Claim~\ref{clm:singlestate}}\\
	&=\asem{S^{\rho_R};S}(a) \hspace{3.5em} \mbox{$\asem{\cdot}$ is the best abstraction}\\
	&=\asem{S}\circ \asem{S^{\rho_R}}(a) \hspace{9.5em} \mbox{Step (i)}\\
	&=\asem{S} \circ \alpha\circ \sem{S^{\rho_R}} \circ \gamma(a) \\
	&=\asem{S} \circ \alpha(\rho_R)   \hspace{10.5em} \mbox{Claim~\ref{clm:singlestate}}\\
	&=\asem{S} \circ \alpha(R).   \hspace{10.8em} \mbox{Claim~\ref{clm:settostate}}
\end{align*}
Thus $\alpha\circ \sem{S} =  \asem{S}\circ\alpha$ as desired. \qedhere
\end{itemize}

\end{proof}

\begin{example}[Abstract interpretation induced by applied quantum Hoare logic]\label{exm:aibyahl}
	Recall from Example~\ref{exm:applied} that the applied quantum Hoare logic~\cite{zhou2019applied}, which is both sound and relatively complete, uses elements of $\shv$ as assertions. Thus by Theorems~\ref{thm:hltoai}, the induced abstract interpretation on subspace domain is complete, where the abstraction and concretisation functions are defined in Lemma~\ref{lem:wsasserttoad} and the abstract operators defined in Eq.\eqref{eq:defasem}. It is easy to check that this induced abstract interpretation for quantum programs is exactly the one defined in Example~\ref{exm:soiscomplete}.
\end{example}	

\begin{example}\label{exm:noexisthl} Let us revisit Example~\ref{exm:yuai}. We have already shown that although the  local-subspace abstract domain $\shv_\sigma $ is well-structured for any signature $\sigma$, 
	the best abstraction $\sem{S}^b_\sigma$ is in general not complete for $\sem{S}$.
	Thus from Theorem~\ref{thm:hltoai}, it is impossible to have a sound and relatively complete Hoare-type logic system with elements in $\shv_\sigma$ being taken as assertions. 
\end{example}

\subsection{Abstract interpretation induced by incorrectness logic}
\label{subsec:iltoai}
Let $\a$ be a well-structured set of quantum assertions, and $\is$ an incorrectness logic system for quantum programs, where the assertions are taken from $\a$. Similar to the case of Hoare logic in the last subsection, we assume that $\is$ consists of an axiom (schema) for each basic command and a proof rule for each program construct, except that the consequence rule is now of the form
\[
\displaystyle\frac{\qassert\leq_\a \qassert',\ \iass{\qassert}{S}{\qassertp},\ \qassertp'\leq_\a \qassertp}{\iass{\qassert'}{S}{\qassertp'}}
\]
for precondition weakening and postcondition strengthening. 
A correctness formula $\iass{\qassert}{S}{\qassert'}$ in $\is$ is semantically valid, written $\models_{\is} \iass{a}{S}{a'}$,  if
\begin{equation}\label{eq:defin}
\qassert' \leq_{\a} \bigwedge \{b\in \a : \forall \rho\models \qassert, \sem{S}(\rho)\models b\}.	
\end{equation}
Note that the right-hand side of Eq.\eqref{eq:defin} denotes the strongest assertion which is satisfied by all reachable states starting from some state satisfying $a$, and $\qassert'$ provides an under-approximation for it. The formula $\iass{\qassert}{S}{\qassert'}$ is derivable, written $\vdash_{\is} \iass{\qassert}{S}{\qassert'}$, if it has a proof sequence in $\is$. 

An abstract interpretation for quantum programs is then naturally induced by $\is$ as follows:
\begin{itemize}
	\item Take $\a$ to be the abstract domain of quantum states, and define the pair of abstraction-concretisation functions $(\alpha, \gamma)$ as in Lemma~\ref{lem:wsasserttoad}. Then $\a$ as an abstract domain is also well-structured.
	\item For any quantum program $S$, let $\asem{S}: \a\ra \a$ with
	\begin{equation}\label{eq:definasem}
		\asem{S}(a) = \bigvee \{a'\in \a :\ \vdash_{\is} \iass{a}{S}{a'}\}
	\end{equation}
	for any $a\in \a$ be the abstract operator of $\sem{S}$.
\end{itemize}

For any quantum program $S$ and $a\in \a$, let
\begin{equation}\label{eq:defwpc}
	wpc(S,a) \define \bigvee \{a'\in \a :\ \models_{\is} \iass{a}{S}{a'}\}.
\end{equation}
We now prove that $wpc(S,a)$ is the weakest postcondition of $a$ with respect to $S$, and it is exactly the best abstraction of $S$.
\begin{lemma}\label{lem:wpost}
	For any quantum program $S$ and $a\in \a$,
	\begin{enumerate}
		\item $\models_{\is} \iass{a}{S}{wpc(S,a)}$;
		\item for any $a'\in \a$, if $\models_{\is} \iass{a}{S}{a'}$ then $a'\leq_\a wpc(S,a)$;
		\item $wpc(S,a) = \bsem{S}(a)$. 
	\end{enumerate}
\end{lemma}
\begin{proof}
	Note that from the definitions of $\alpha$ and $\gamma$ in Lemma~\ref{lem:wsasserttoad}, we have
	\begin{align*}
		\bsem{S}(a) & = \alpha\circ\sem{S}\circ\gamma(a)\\
		& = \bigwedge\left\{b\in \a: \sem{S}\circ\gamma(a)\subseteq \gamma(b)\right\} \\
		& = \bigwedge\left\{b\in \a:  \forall \rho\in \gamma(\qassert), \sem{S}(\rho)\in \gamma(b)\right\} \\
		& = \bigwedge\left\{b\in \a:  \forall \rho\models \qassert, \sem{S}(\rho)\models b\right\}.
	\end{align*}
	Thus
	 $\models_{\is} \iass{a}{S}{a'}$ iff $a'\leq_{\a} \bsem{S}(a)$, and so $wpc(S,a) = \bsem{S}(a)$ from Eq.~\eqref{eq:defwpc}. Finally, (1) and (2) follow from (3) directly.
\end{proof}

The following theorem gives a close relationship between a quantum incorrectness logic system and the abstract interpretation induced by it.

\begin{theorem}\label{thm:iltoai}
	Let $\a$ be a well-structured set of quantum assertions, and $\is$ an  incorrectness logic system for quantum programs with assertions taken from $\a$. If $\is$ is sound and relatively complete, then the induced abstract interpretation is complete.
\end{theorem}
\begin{proof}
	For any $a\in \a$ and quantum program $S$, let 
	$$G_{\vdash}(S,a) \define \{a'\in \a :\ \vdash_{\is} \iass{a}{S}{a'}\}\quad \mbox{and}\quad 
	G_{\models}(S,a) \define \{a'\in \a :\ \models_{\is} \iass{a}{S}{a'}\}.$$ 
	As $\is$ is sound and relatively complete, we have $G_{\vdash}(S,a) = G_{\models}(S,a)$. Thus $\asem{S}(a) = \bsem{S}(a)$ from Lemma~\ref{lem:wpost}(3), and so $\asem{S}$ is the best abstraction of $S$. 
	
	Let $S_1$ and $S_2$ be two quantum programs.
	From Lemma~\ref{lem:wpost} and completeness of $\is$, we have for all $a\in \a$,
			\[
	\vdash_{\is}\iass{a}{S_1}{\asem{S_1}(a)}\quad \mbox{and}\quad \vdash_{\is} \iass{\asem{S_1}(a)}{S_2}{\asem{S_2}\circ \asem{S_1}(a)},
	\]
	and so $\vdash_{\is}\iass{a}{S_1;S_2}{\asem{S_2}\circ \asem{S_1}(a)}$ from the rule for sequential composition. Thus
		$\asem{S_2}\circ \asem{S_1}(a)\leq_{\a} \asem{S_1;S_2}(a)$ from Eq.~\eqref{eq:definasem}. However, from the fact that $\gamma\circ \alpha\geq_{\qconc} \mathrm{id}_{\qconc}$, we know that $\asem{S_1; S_2}(a)\leq_\a \asem{S_2}\circ \asem{S_1}(a)$. Thus
		\[
		\asem{S_1; S_2} = \asem{S_2}\circ \asem{S_1}.
		\]
The rest of the proof is the same as in Theorem~\ref{thm:hltoai} (starting from the second step).
\end{proof}

	\begin{example}[Abstract interpretation induced by quantum incorrectness logic]
		Recall from Example~\ref{exm:inclogic} that the quantum incorrectness logic presented in~\cite{yan2022incorrectness} is both sound and relatively complete. Thus by Theorems~\ref{thm:iltoai}, the induced abstract interpretation on subspace domain is complete, where the abstraction and concretisation functions are defined in Lemma~\ref{lem:wsasserttoad} and the abstract operators defined in Eq.\eqref{eq:definasem}. It is easy to check that this induced abstract interpretation for quantum programs is also the one defined in Example~\ref{exm:soiscomplete}.
	\end{example}	

Similar to Example~\ref{exm:noexisthl}, Theorem~\ref{thm:iltoai} also implies that it is impossible to have a sound and relatively complete incorrectness logic system with elements in $\shv_\sigma$ being taken as assertions. 

To conclude this section, we note the following corollary, which can be directly shown from Theorems~\ref{thm:aitohl}, \ref{thm:aitoihl}, \ref{thm:hltoai}, and \ref{thm:iltoai}.
	\begin{corollary}
		Let $\a$ be a well-structured set of assertions for quantum states. The following two statements are equivalent:
		\begin{enumerate}
			\item there exists a sound and relatively complete quantum Hoare logic system; 
			\item there exists a sound and relatively complete quantum incorrectness logic system.
		\end{enumerate} 
	\end{corollary}

\section{Conclusion}
\label{sec:con}

We have shown a close relationship between abstract interpretation and Hoare/incorrectness logic for quantum programs, when the abstract domain and the set of assertions for quantum states are well-structured. With this relationship, we obtain sound and relatively complete Hoare logic and incorrectness logic for quantum programs. The induced logic systems are in a forward manner, complementing the (backward) applied quantum Hoare logic and incorrectness logic proposed in the literature. Conversely, our result also implies the non-existence of any sound and relatively complete Hoare or incorrectness logic for quantum programs if tuples of local subspaces are taken as assertions for quantum states.  
 
For future work, we plan to consider quantitative assertions where a quantum state satisfies a property with some degree (a number in $[0,1]$). Natural candidates for such assertions are hermitian operators between 0 and the identity, as widely used in the expectation-based quantum Hoare logics. To this end, we have to first establish a theory of abstract interpretation when such hermitian-operator assertions are regarded as abstraction of quantum states.

		\bibliography{ref}

\end{document}